\documentclass[12pt]{article}
\usepackage[margin=1in]{geometry} 
\usepackage{amsmath,amsthm,amssymb,hyperref}
\usepackage{mdframed,tcolorbox}
\usepackage{float}
\usepackage{graphicx,psfrag,epsf, subcaption, enumitem,multirow}
\usepackage{mathrsfs}
\usepackage[round, sort]{natbib}
\usepackage{longtable}
\usepackage{setspace}
\usepackage{lscape}
\usepackage{algorithm}
\usepackage{algorithmic}
\newtheorem{remark}{Remark}

\newtheorem{definition}{Definition}
\newtheorem{lemma}{Lemma}
\newtheorem{theorem}{Theorem}

\newtheorem{assumption}{Assumption}
\usepackage{xcolor}

\def\##1\#{\begin{align}#1\end{align}}
\def\$#1\${\begin{align*}#1\end{align*}}

\newcommand{\cN}{\mathcal{N}}

\newcommand{\CC}{\mathbb{C}}
\newcommand{\EE}{\mathbb{E}}

\newcommand{\NN}{\mathbb{N}}
\newcommand{\PP}{\mathbb{P}}

\newcommand{\RR}{\mathbb{R}}

\allowdisplaybreaks[4]
\numberwithin{equation}{section}

\title{Testing for large-dimensional covariance matrix under differential privacy}

 \author{
Shiwei Sang\thanks{Shiwei Sang and Yicheng Zeng are co-first authors of this work.}
\thanks{School of Mathematics and Statistics, Xi’an Jiaotong University, Xi'an, China.}
\and
Yicheng Zeng\footnotemark[1]
\thanks{School of Science, Sun Yat-sen University, Shenzhen, China.}
\and
Xuehu Zhu\thanks{Xuehu Zhu and Shurong Zheng are co-corresponding authors.}
\footnotemark[2]
\and
Shurong Zheng\footnotemark[4]
\thanks{School of Mathematics and Statistics, Northeast Normal University, Changchun, China.}
}

\begin{document}

\maketitle
	
\begin{abstract}

The increasing prevalence of high-dimensional data across various applications has raised significant privacy concerns in statistical inference. In this paper, we propose a differentially private integrated test statistic for testing large-dimensional covariance structures, enabling accurate statistical insights while safeguarding privacy. First, we analyze the global sensitivity of sample eigenvalues for sub-Gaussian populations, where our method bypasses the commonly assumed boundedness of data covariates. For sufficiently large sample size, the privatized statistic guarantees privacy with high probability. Furthermore, when the ratio of dimension to sample size, $d/n \to y \in (0, \infty)$, the privatized test is asymptotically distribution-free with well-known critical values, and detects the local alternative hypotheses distinct from the null 
at the fastest rate of $1/\sqrt{n}$. Extensive numerical studies on synthetic and real data showcase the validity and powerfulness of our proposed method. 


\end{abstract}
\newpage
\tableofcontents

\newpage

\section{Introduction}
	
The covariance matrix plays a central role in multivariate statistical analysis and has long been regarded the cornerstone of many fundamental statistical methods~\citep{0An,Hardle2019, mertler2021advanced}. With the emergence of high-dimensional data in a wide range of contemporary applications, such as genomics, bio-informatics, image processing, financial economics and social sciences~\citep{fan2014challenges}, statistical inference for large-dimensional covariance structures has gradually become one of the most prominent research highlights.
	
With advances in machine learning and data mining, there has been an increasing availability of large datasets that contain sensitive information. Conventional statistical inference emphasizes accuracy, but often neglects individual privacy, thereby increasing the risk of personal information disclosure. \textit{Differential privacy} (DP), initially proposed in~\citep{dwork2006calibrating, dwork2006our}, has gradually become the most widely adopted privacy criterion. Their seminal work provides a general and flexible framework for privacy problems and has become a commonly used standard in many realistic applications~\citep{erlingsson2014rappor, ding2017collecting, team2017learning}. In recent years, DP-based inference has drawn significant attention, enabling accurate statistical insights while making it challenging for attackers to infer personal information from published statistics.

Although much of research has focused on estimation problems~\citep{cai2024optimal,cai2021cost,amin2019differentially,liu2022dp}, hypothesis testing within the differential privacy framework, especially in high-dimensional settings, remains underexplored~\citep{narayanan2022private, canonne2020private, huang2022comparison}.  
To the best of our knowledge, DP-based testing for large-dimensional covariance structures has not yet been investigated. Driven by this research gap, this paper focuses on testing the large-dimensional covariance matrix under differential privacy constraints. 
More specifically, we investigate the following simple global hypothesis \eqref{simple hypothesis} under DP via independent samples $\{\mathbf{x}_i\}_{i = 1}^n$ drawn from the population $\mathbf{x}$:
\begin{equation}\label{simple hypothesis}
H_0: \mathbf{\Sigma} = \mathbf{I}_d \quad {\rm versus} \quad H_1: \mathbf{\Sigma} \neq \mathbf{I}_d,
\end{equation}
where the dimension $d$ and sample size $n$ are comparable in the sense that $y_n:= d/n \to y \in (0,\infty)$. Here $\mathbf{\Sigma}$ and $\mathbf{I}_d$ denote the population covariance matrix and identity matrix, respectively. It is worth noting that testing $H_0: \mathbf{\Sigma} = \mathbf{\Sigma}_0$ against $H_1: \mathbf{\Sigma} \neq \mathbf{\Sigma}_0$ for some given $\mathbf{\Sigma}_0$ can always be reduced to problem \eqref{simple hypothesis} by the transformation $\mathbf x \mapsto \mathbf{\Sigma}_0^{-1/2}\mathbf{x}$.
   
In the classical literature of low dimension and large sample size, many testing methods for covariance matrix have been developed in multivariate statistics~\citep{0An,Hardle2019,mertler2021advanced}. For instance, under the normality assumption, a frequently-used statistic for testing $H_0: \mathbf{\Sigma} = \mathbf{I}_d$ is the likelihood ratio (LR) statistic, which is defined by
\begin{equation}\label{LRT}
L_{n, d}:=n \cdot \sum_{1\leq i \leq d}\left\{{\lambda}_i - \log(\lambda_i)-1\right\},
\end{equation}
where $\{\lambda_i\}_{i = 1}^d$ are eigenvalues of the sample covariance matrix $\mathbf{S_x}:= n^{-1}\sum_{i=1}^n (\mathbf{x}_i - \bar{\mathbf{x}})(\mathbf{x}_i - \bar{\mathbf{x}})^{\top}$ with $\bar{\mathbf{x}}: = n^{-1}\sum_{i=1}^n \mathbf{x}_i$ denoting the sample mean. It is known that $L_{n,d}$ weakly converges to the $\chi^2_{d(d+1)/2}$ distribution as $n \to \infty$ under $H_0$. However, almost all conventional testing methods, including likelihood ratio tests, fail when the dimension grows with the sample size~\citep{ledoit2002some}. To deal with the high-dimensional issue,~\citep{bai2009corrections} proposed a corrected likelihood ratio (CLR) statistic, which is defined as
\begin{equation}\label{CLR}
\widetilde{L}_{n,d}:= \sum_{1\leq i \leq d}\left\{{\lambda}_i - \log(\lambda_i)-1 - F^{y_n}\right\},
\end{equation}
where $F^{y_n} := 1+(1-y_n)\log (1-y_n)/y_n$. With the help of techniques from modern random matrix theory (RMT), \citep{bai2009corrections} derived its limiting null distribution when $y_n \to y \in (0,1)$ without normality assumption, thus yielding a valid test.

In the literature of DP-based hypothesis testing, it is common practice to use randomized mechanisms to extend non-private tests to their private counterparts. It is worth noting that either the LR statistic \eqref{LRT} or its corrected version \eqref{CLR} only relies on the eigenvalues of $\mathbf{S_x}$ without requiring any additional information. Hence, given the dataset, the statistics for testing \eqref{simple hypothesis} can be obtained by performing a deterministic query, with the sample eigenvalues as the output. Inspired by this, we propose a privatized test statistic for  problem \eqref{simple hypothesis}, which is given by
\begin{equation}\label{LRTdp}
L_{1}^{\text{dp}}:= \frac{1}{K} \cdot \sum_{1\leq i\leq K}\left(\left|\tilde\lambda_i\right| - \log \left| \tilde \lambda_i\right|-1\right),
\end{equation}
where $K: = \min\{d,n\}$ and $\{\tilde \lambda_i\}_{i=1}^K$ are the privatized eigenvalues of $\mathbf{S_x}$.  

Under the classical parametric testing framework, the likelihood ratio test is widely regarded as the most powerful test according to the \textit{Neyman-Pearson} lemma~\citep{lehmann2008testing}. However, this optimality relies on specific assumptions, such as fully known distributions and fixed dimensions. In high-dimensional settings, the parameter space increases rapidly, and the underlying distributions may be partially or entirely unspecified, making the conditions of the Neyman-Pearson lemma inapplicable. Furthermore, privacy constraints often require adding noise or applying other modifications to the data, which can distort the likelihood functions and undermine the optimality of the likelihood ratio test. To address these challenges and enhance the reliability of test results, we complement the statistic \eqref{LRTdp} with two additional privatized statistics, based on quadratic loss and absolute deviation loss, defined as
\begin{equation}\label{L2dp} L_{2}^{\mathrm{dp}} := \frac{1}{K} \cdot \sum_{1 \leq i \leq K} \left|\tilde{\lambda}_i - 1\right|^2,  \end{equation}
and
\begin{equation}\label{L3dp} L_{3}^{\mathrm{dp}} := \frac{1}{K} \cdot \sum_{1 \leq i \leq K} \left|\tilde{\lambda}_i - 1\right|,
\end{equation}
respectively. \citep{zheng2019hypothesis} indicates that the quadratic loss test appears to have a higher power than the likelihood ratio test in specific settings. This current work proposes an integrated strategy that combines the strengths of multiple loss functions. This integrated statistic allows for a more comprehensive assessment of deviations under alternatives, and presents the robustness of the testing procedure in high-dimensional and privacy-constrained scenarios.

In the design of the privatized statistics (\ref{LRTdp}-(\ref{L3dp}), the privatization for sample eigenvalues can be achieved by some classical DP mechanisms, including \textit{Laplacian} mechanism and \textit{Gaussian} mechanism~\citep{dwork2014algorithmic}. 
In this process, the so-called \textit{global sensitivity} plays a crucial role, as it determines the scale of noise injected into the original sample eigenvalues, which is essential for both privacy guarantees and the accuracy of statistical inference. While most existing work concentrated on testing for discrete or bounded populations~\citep{ding2018comparing,dunsche2022multivariate,huang2022comparison,gaboardi2016differentially}, which ensures a finite noise scale, we analyze the global sensitivity of sample eigenvalues for unbounded data. Our theoretical results confirm that when the sample size is sufficiently large, the privatized statistics (\ref{LRTdp})-(\ref{L3dp}) ensure privacy protection with high probability. Given a properly calibrated noise scale, we can establish a general framework for testing covariance structures while preserving differential privacy.

\subsection{Related work and discussion}
In recent years, there has been extensive research on hypothesis testing under privacy constraints. Under the framework of differential privacy, numerous privatized tests have been proposed as options to classical inference methods in order to address privacy concerns, such as mean tests~\citep{ding2018comparing, dunsche2022multivariate}, tests for distributions~\citep{lam2022minimax,dubois2019goodness, canonne2019structure},  tests for independence~\citep{wang2015revisiting, gaboardi2016differentially,couch2019differentially}, tests in linear regression models~\citep{sheffet2017differentially,wang2018revisiting, alabi2022differentially, alabi2022hypothesis}, and some kernel-based two-sample tests~\citep{ kim2023differentially, raj2019differentially}. 
However, no testing procedure for covariance structures under DP is provided therein. To bridge this gap, a natural approach is to introduce perturbations to existing tests using classical DP mechanisms. This, however, raises two key challenges. First, many testing methods for large-dimensional matrix have been developed, which types of statistics should be chosen as the benchmark? Second, in the high-dimensional and unbounded settings, how can we make the trade-off between privacy assurance and statistical accuracy?

In the literature of RMT, many testing methods for large-dimensional covariance matrix have been developed~\citep{srivastava2005some, srivastava2012testing, birke2005note,cai2013optimal,zheng2015substitution,zheng2019hypothesis}. For testing whether a covariance matrix equals the identity matrix in (\ref{simple hypothesis}), there has been extensive research based on the spectrum of the sample covariance matrix. This problem  was studied in~\citep{johnstone2001distribution} focusing on spiked-type alternatives, where they particularly addressed the behavior of the largest eigenvalue of sample covariance matrix under $H_0: \mathbf{\Sigma} = \mathbf{I}_d$. Under normality assumption, the distribution of the largest eigenvalue converges to the Tracy-Widom law as $d/n \to y \in (0,\infty)$. This work was later considered in sub-Gaussian cases~\citep{soshnikov2002note} and further generalized to  broader settings with moment conditions~\citep{peche2009universality}. While LR statistics are no longer applicable as the dimension grows with the sample size, \citep{bai2009corrections} theoretically analyzed the failure of the LR statistic \eqref{LRT} and proposed the corrected version \eqref{CLR}, deriving its asymptotic null distribution with some moment conditions. Their work was further studied in~\citep{wang2013identity, jiang2012likelihood,zheng2015substitution}, and has been adapted for contemporary applications, such as cognitive radio networks~\citep{taherpour2024large} and multimodal imaging data~\citep{chang2024statistical}. These spectrum-based statistics solely rely on sample eigenvalues, thereby simplifying privacy analysis. In parallel, the eigenvectors can be safeguarded using even more secure methods, which ensures  comprehensive  privacy protection across the dataset. With this dual focus on eigenvalue analysis and enhanced eigenvector security, we construct the privatized statistics (\ref{LRTdp})-(\ref{L3dp}) with the aid of classical DP mechanisms.

To obtain accurate statistical inference while preserving privacy, the amount of the introduced noise must be meticulously controlled. There are extensive studies concentrating on bounded or truncated data for other statistical tasks~\citep{dunsche2022multivariate, amin2019differentially, dwork2014analyze, chaudhuri2011differentially, liu2022dp, blum2005practical,ding2018comparing}, which ensures finite global sensitivity, and thus simplifies the determination of the noise scale. Inspired by some recent studies~\citep{cai2024optimal, cai2021cost}, this work characterizes a probabilistic upper bound for the global sensitivity of sample eigenvalues in high-dimensional sub-Gaussian regime, allowing our proposed tests to adapt to a wide range of data scenarios.

\subsection{Organization}
The rest of this paper is organized as follows. Section \ref{section:2} provides a brief overview of differential privacy, along with some relevant concepts and results from random matrix theory. In Section \ref{section:3}, we develop the privatized test, and study its privacy guarantees and its asymptotic properties.  Section \ref{section:4} presents simulation studies and real-data applications to illustrate the performance of our proposed method.  Section \ref{section:5} concludes with a discussion of the method's strengths and limitations, as well as potential directions for future research. Proofs of our main results and technical lemmas are provided in the Appendices \ref{app: A}-\ref{app: B}.

\section{Notation and preliminaries}
\label{section:2}
\subsection{Notation}
Throughout this paper, we use $c,c^{\prime}, c_1,c_2$ to denote universal constants, which may vary from line to line. For a vector $\mathbf{x}$, $\mathbf{x}^{\top}
$ and $ \|\mathbf{x}\|_p$ denote the transpose and $\ell_p$ norm of $\mathbf{x}$, respectively. For a $d\times d$ Hermitian matrix $\mathbf{A}$, $\mathbf{A}^{\top}, \mathrm{Tr}(\mathbf{A})$, and $\|\mathbf{A}\|$ denote the transpose, trace and spectral norm of $\mathbf{A}$, respectively. Let $\lambda_1(\mathbf{A}) \geq \cdots \geq \lambda_d(\mathbf{A})$  denote the ordered eigenvalues of 
$\mathbf{A}$. For a random variable (or vector) $X$, $\mathbb E(X), \mathrm{Var}(X)$, and $ \mathrm{Cov}(X)$ denote the expectation, variance and covariance of $X$, respectively. The notation ``$\xrightarrow{ D}$" denotes  convergence in distribution. The notation ``$\mathrm{Lap}(\mu,\sigma)$'' denotes  Laplacian distribution with the location parameter $\mu$ and the scale parameter $\sigma$. The indicator function of an event $\mathcal{A}$ is denoted as $\mathbf{1}_{\{\mathcal{A}\}}$. The cardinality of a set $S$ is denoted as $\sharp_{\{S\}}$. 

For two quantities $a$ and $b$, $a\vee b \equiv \max\{a,b\}$ denotes the larger of the two, while $a \wedge b \equiv \min\{a,b\}$ denotes the smaller. For a random sequence $\{X_n\}_{n\geq1}$ and deterministic sequences $\{a_n\}_{n\geq1}, \{b_n\}_{n\geq1}$, $X_n = O_p(a_n)$
means $X_n/a_n$ is bounded in probability and $X_n = o_p(a_n)$
means $X_n/a_n$ converges to zero in probability. Furthermore, $a_n \lesssim b_n$ (or $a_n = O(b_n)$) means $|a_n| \leq c_1|b_n|$ for some constant $c_1>0$,   $a_n \gtrsim b_n$ means $|a_n| \geq c_2|b_n|$ for some constant $c_2>0$, $a_n \asymp b_n$ means $c_1 |b_n| \leq |a_n| \leq c_2 |b_n|$ for some constants $c_1,c_2 >0$, and $a_n = o(b_n)$ means $\lim_{n \to \infty} a_n/b_n =0$. 

\subsection{Review on differential privacy}
This subsection provides a brief overview of fundamental concepts and properties of differential privacy. A matrix $\mathbf{X}\in \RR^{n\times d}$ consisting of $n$ samples from a $d$-dimensional population $\mathbf{x}$ is called a dataset (or design matrix). Two datasets $\mathbf{X}$ and $\tilde{\mathbf{X}}$ are called \textit{neighboring} if they differ by only one sample, denoted as $\mathbf{X} \simeq \tilde{\mathbf{X}}$. 
A \textit{query} $f: \mathbb R^{n \times d} \to \mathcal R$ is a function that maps the dataset to a specified output space $\mathcal{R}$, while a \textit{randomized mechanism} $\mathcal M:\mathbb R^{n \times d} \to \mathcal R$ provides an approximation of the query result. The definition of differential privacy is as follows.
     
     \begin{definition}[Differential privacy~\citep{dwork2014algorithmic}]
     	\label{def:dp}
     	A randomized mechanism $\mathcal{M}:\mathbb{R}^{n \times d} \to \mathcal{R}$ is called $(\varepsilon, \delta)$-differentially private (($\varepsilon,\delta$)-DP) if for any $\varepsilon >0,\delta \in [0,1)$, the inequality
     	\begin{equation}
     		\mathbb{P}\left\{\mathcal{M}(\mathbf{X}) \in S\right\} \leq e^{\varepsilon} \cdot \mathbb{P}\left\{\mathcal{M}(\tilde{\mathbf{X}}) \in S\right\} +\delta
     	\end{equation}
      	 holds for any neighboring datasets $\mathbf{X}\simeq \tilde{\mathbf{X}}$ and any measurable subset $S \subset \mathcal{R}$.
     \end{definition}
     
In the above definition, $\varepsilon$ and $\delta$ are parameters that quantify privacy loss, with smaller values indicating less privacy loss and stronger privacy guarantees. As stated in the~\citep{dwork2014algorithmic}, $\varepsilon$ is usually considered as a small constant, and $\delta$ is ideally set to be less than the reciprocal of any polynomial in the sample size.  As the sample size grows, the possibility of privacy breaches diminishes rapidly, which maintains a high level of privacy guarantees, particularly in large datasets. 
When $\delta = 0$, $(\varepsilon,0)$-differential privacy is simply called $\varepsilon$-DP (or pure-DP). Intuitively, $(\varepsilon, \delta)$-DP provides an approximation of $\varepsilon$-DP, allowing the randomized mechanism to preserve $\varepsilon$-DP with high probability about $1-\delta$.
     
Several mechanisms for achieving DP have been extensively studied~\citep{dwork2014algorithmic}, 
with the most common approach being the addition of random noise to query results. A key consideration in DP research is determining the appropriate noise scale, which plays a crucial role in balancing privacy protection and data utility. This noise scale is directly linked to the concept of \textit{global sensitivity}, which quantifies the maximum possible change in a query's output when applied to neighboring datasets. Formally, the global sensitivity of $f$ is defined by
     \begin{equation}\label{sensitivity}
     	\Delta f := \sup_{\mathbf{X} \simeq \tilde{\mathbf{X}}} \left\|f(\mathbf{X})-f(\tilde{\mathbf{X}})\right \|,
     \end{equation}
Especially, when $\| \cdot \|$ represents the $\ell_1$ norm defined on $\mathcal{R}=\mathbb{R}^s$ with $s \geq 1$, the $\ell_1$ sensitivity of $f$, denoted as $\Delta_1f$, determines the noise scale of \textit{Laplacian} mechanism.

     \begin{lemma}[Laplacian mechanism~\citep{dwork2014algorithmic}]\label{lemma Laplacian mechanism}
     	Assume that $f:  \mathbb{R}^{n \times d} \to \mathbb{R}^s$ is a $s$-dimensional vector-valued query. Given $\varepsilon>0$, the Laplacian mechanism is defined by
     	\begin{equation}
     		\label{laplacian mechanism}
     		\mathcal{M}_{\mathrm{Lap}}(\mathbf{X}):=f(\mathbf{X})+(\ell_1,\cdots,\ell_s)^{\top},
     	\end{equation}
     	where $\{\ell_i\}_{i = 1}^s$ are independently drawn from the central Laplacian distribution with scale parameter of $\Delta_1 f /\varepsilon$. Furthermore, $\mathcal{M}_{\mathrm{Lap}}$ is $\varepsilon$-differentially private.
     \end{lemma}

     \begin{remark}[Gaussian mechanism]
         To achieve $(\varepsilon,\delta)$-DP, one can use the Gaussian mechanism~\citep{dwork2014algorithmic,dwork2006our}, which is another frequently-used mechanism and applies Gaussian noise scaled to the $\ell_2$ sensitivity. As shown in~\citep{acharya2018differentially}, any algorithm preserving $(\varepsilon + \delta,0)$-DP also satisfies $(\varepsilon,\delta)$-DP, while the former can be achieved straightforwardly using Laplacian mechanism. Therefore, for the reminder of the paper, we focus on testing procedures implemented via the Laplacian mechanism.
     \end{remark}

In practical applications, private algorithms are often designed in a modular manner. An appealing feature of DP is its \textit{post-processing} property, which allows complex private algorithms to be built from simpler ones without compromising privacy guarantees.

     \begin{lemma}[Post-processing~\citep{dwork2014algorithmic}]
     	\label{Post-processing}
     	Assume that $\mathcal{M}: \mathbb{R}^{n \times d } \to \mathcal{R}$ is $(\varepsilon,\delta)$-differentially private and $\phi:\mathcal{R} \to \mathcal{R}^{\prime}$ is a measurable mapping from $\mathcal{R}$ to $\mathcal{R}^{\prime}$, then the mechanism $\phi \circ \mathcal{M}:\mathbb{R}^{n \times d} \to \mathcal{R}^{\prime}$ remains $(\varepsilon,\delta)$-differentially private.
     \end{lemma}

\subsection{Preliminaries in random matrix theory}
     
In this subsection, we introduce some concepts and results from random matrix theory. 
     
\begin{definition}[Empirical spectral distribution, ESD]\label{ESD}
For any $d \times d$ Hermitian matrix $\mathbf{A}$ with eigenvalues $\lambda_1(\mathbf{A}) \geq \cdots \ge \lambda_d(\mathbf{A})$, the empirical spectral distribution of $\mathbf{A}$ is defined as
\begin{equation}\label{esd}
F_{d}^{\mathbf{A}}(t):= \frac{1}{d}\sum_{i=1}^d \mathbf{1}_{\{\lambda_i(\mathbf{A}) \leq t\}}.
\end{equation}
\end{definition}
     
\begin{definition}[Limiting spectral distribution, LSD]\label{LSD}
For any $d \times d$ Hermitian matrix $\mathbf{A}  $ with empirical spectral distribution $F_{d}^\mathbf{A}(t)$, if $F_d^\mathbf{A}(t)$ weakly converges to a measure as $d \to \infty$, denoted as $F^\mathbf{A}(t)$, then the measure $F^\mathbf{A}(t)$ is called the limiting spectral distribution of $\mathbf{A}$.
\end{definition}

A fundamental result~\citep{marchenko1967distribution} in random matrix theory tells us that for a normalized population $\mathbf{z}$, the LSD of the sample covariance matrix $
\mathbf{S_z} =n^{-1} \sum_{i=1}^n(\mathbf{z}_i-\bar{\mathbf{z}})(\mathbf{z}_i-\bar{\mathbf{z}})^{\top}$ is the well-known Marcenko-Pastur (M-P) law in the high-dimensional regime $y_n \to y \in (0,\infty)$. Here, the M-P law $F_{y}(t)$ has a density function
\begin{equation}
\label{MP density}
f_{y}(t) = \frac{1}{2\pi yt} \sqrt{(t-\lambda_{-})(\lambda_{+}-t)} \cdot \mathbf{1}_{\{\lambda_{-} \leq t\leq \lambda_{+} \}}
\end{equation}
and has a point mass $1-1/y$ at the origin if $y >1$, where $\lambda_{-} = (1-\sqrt{y})^2$ and $\lambda_{+} = (1+\sqrt{y})^2$.

\begin{lemma}[Marcenko-Pastur Law; Theorem~2.9 in~\citep{yao2015sample}]\label{mp_law_1}
Suppose that the entries of the design matrix $\mathbf{Z}$ from $\mathbf{z}$ are independently and identically distributed (i.i.d.) random variables with mean zero and unit variance, and $y_n \to y \in (0,\infty)$. Then almost
surely, $F_d^{\mathbf{S_z}}$ weakly converges to the M-P law defined by (\ref{MP density}).
\end{lemma}

The Stieltjes transform provides a convenient way to represent and manipulate probability distributions, particularly in the context of RMT.

\begin{definition}[Stieltjes transform]\label{Stieltjes transform}
Let $P$ be a probability measure on $\mathbb R$. Its Stieltjes transform is defined by
\begin{align*}
m_P(z) := \int \frac{1}{x-z} \mathrm{d}P(x),\quad z \in \mathbb{C} \setminus \Gamma_{P},
\end{align*}
where $\Gamma_P$ denotes the support of $P$.
\end{definition}

Under certain conditions, one can reconstitute $P$ from its Stieltjes transformation $m_P(z)$ by the {\it inverse formula}. For any two continuity points $a<b$ of $P$, it holds that 
\#\label{inverse_formula}
P([a,b])=\lim_{\nu\rightarrow 0+}\frac{1}{\pi}\int_a^b {\rm Im}\ m_P(x+\mathbf{i} \nu)\mathrm{d}x,
\#
where $\mathrm{Im}\, z$ denotes the imaginary part of $z$, and $\mathbf{i}$ is the imaginary unit.

With these notions, the M-P law can be generalized to more general settings as follows.
\begin{lemma}[Generalized Marcenko-Pastur Law; Theorem~2.14 in~\citep{yao2015sample}]\label{Generalized_MP}
Consider the sample covariance matrix $\mathbf{S_x}=n^{-1}\mathbf{\Sigma}^{1/2}\mathbf Z^\top \mathbf Z\mathbf{\Sigma}^{1/2}$ for some positive semidefinite matrix $\mathbf{\Sigma}$, where $\mathbf{Z}$ is defined in Lemma \ref{mp_law_1}. If the ESD of $\mathbf{\Sigma}$,  denoted by $Q_n$, weakly converges to a nonrandom probability measure $Q$, then almost surely, the ESD of $\mathbf{S_x}$ weakly converges to a nonrandom probability measure $F_{y,Q}$ when $y_n \to y \in (0,\infty)$, whose Stieltjes transform $m(z)$ is determined by the equation
\#\label{Stieltjes_equation_general_Sigma}
m_{y,Q}(z)\equiv m(z)=\int\frac{1}{t\{1-y-yzm(z)\}-z}\mathrm{d} Q(t),\quad z\in \CC^+.
\#
\end{lemma}

\section{Privatized covariance matrix testing}
    \label{section:3} 
In this section, we present the privatized test tailored for the high-dimensional testing problem (\ref{simple hypothesis}). We begin with some assumptions for the rest of the paper.
     
\begin{assumption}[High dimensions]\label{high dimension assumption}
Assume that the dimension $d$ and the sample size $n$ both tend to infinity such that $y_n \to y$ for some constant $y \in (0,\infty)$. 
\end{assumption}

\begin{assumption}[Sub-Gaussian data]\label{sub gaussian assumption}
There exists some $d \times d$ positive semidefinite matrix $\mathbf{\Sigma}$ such that $\mathbf{x}=\mathbf{\Sigma}^{1/2}\mathbf{z}$, where $\mathbf{z}$ follows a sub-Gaussian distribution such that $\mathbb E(\mathbf{z})=\mathbf{0},{\rm Cov}(\mathbf{z}) = \mathbf{I}_d$ and
$$
\mathbb E \left\{ \exp\left( \boldsymbol{\alpha}^{\top}\mathbf z \right)\right\} \leq \exp\left( \frac{\|\boldsymbol{\alpha}\|_2^2 \,\sigma^2}{2}\right),\quad \forall\, \boldsymbol{\alpha}\in \mathbb R^d
$$
for some $\sigma>0$. Further, there exists some constants $\gamma >0$ and $c >0$ such that $ \mathrm{Tr}(\mathbf{\Sigma})\leq \gamma d$ and $\|\mathbf{\Sigma}\| \leq c$.
\end{assumption}

\subsection{Sensitivity analysis and privacy guarantee}
In most existing studies on differentially private statistical tasks, the design matrix $\mathbf{X}$ is typically assumed to be fixed, with the $\ell_2$ norm of its columns bounded \citep{amin2019differentially, chaudhuri2011differentially, dwork2014analyze, blum2005practical}. However, this assumption is often too restrictive and impractical, especially for common populations, such as unbounded ones. To address this limitation and identify an appropriate noise scale for achieving DP with data drawn from sub-Gaussian distributions, we begin by analyzing the $\ell_1$ sensitivity of sample eigenvalues in such settings.  
Formally, the query for sample eigenvalues $\Lambda: \mathbb{R}^{n \times d} \to \mathbb{R}^{d}$ takes the design matrix $\mathbf{X}$ as input and outputs the eigenvalues of $\mathbf{S_x}$, whose $\ell_1$ sensitivity is defined as:
\begin{equation}\label{eigenvalue sensitivity}
\Delta_1 \Lambda: = \sup_{\mathbf{X}\simeq \tilde{\mathbf{X}}} \left\|\Lambda\left(\frac{1}{n}\mathbf{X}^{\top}\mathbf{X} \right) - \Lambda\left(\frac{1}{n} \tilde{\mathbf{X}}^{\top} \tilde{\mathbf{X}} \right)\right\|_1.
\end{equation}
The following result  presents a high-probability bound for the $\ell_1$ sensitivity of sample eigenvalues.
\begin{theorem}\label{upper bound of eigenvalue sensitivity}
Under Assumption \ref{sub gaussian assumption}, for any $ t >0 $, we have 
\begin{equation}\label{theorem1}
\Delta_1 \Lambda \leq  \sigma^2 \cdot\left\{\frac{2\mathrm{Tr}(\mathbf{\Sigma})}{n} + 4 \sqrt{ \mathrm{Tr}(\mathbf{\Sigma}^2)} \cdot \sqrt{\frac{t}{n}}+4 \left\|\mathbf{\Sigma}\right\|t   \right\}    
\end{equation}
with the probability at least $1-2 e^{-nt} $, where $\sigma>0$ is the sub-Gaussian parameter of $\mathbf \Sigma^{-1/2}\mathbf x$. 
\end{theorem}
     
\begin{remark}
As shown in the proof of Theorem 1 (see Appendix), the $\ell_1$ sensitivity of $\Lambda$ can be controlled by
\begin{equation}\label{tightest upper bound}
\Delta_1 \Lambda \leq \frac{1}{n} \mathbf{x}_1^{\top}\mathbf{x}_1 + \frac{1}{n} \tilde{\mathbf{x}}_1^{\top}\tilde{\mathbf{x}}_1,
\end{equation}
where $\mathbf{x}_1$ and $\tilde{\mathbf{x}}_1$ are independent copies drawn from $\mathbf{x}$. Due to the unboundedness of $\mathbf{x}$, the standard Laplace mechanism cannot be directly applied, as it requires finite sensitivity to ensure privacy protection.  As a solution, Theorem \ref{upper bound of eigenvalue sensitivity} provides a high-probability upper bound for the sensitivity of sample eigenvalues. Specifically, the inequality in  (\ref{theorem1}) indicates that $\Delta_1\Lambda$  is controlled by a quantity depending on $ \mathbf{\Sigma}, n,\sigma$, and a tuning parameter $t$. By substituting this upper bound for $\Delta_1\Lambda$ into the Laplace mechanism, the statistic (\ref{LRTdp}) achieves $\varepsilon$-DP with probability tending to one at an exponential rate.
\end{remark}

When \(t\) is chosen as a sequence tending to zero, the first term in  (\ref{theorem1}) becomes dominant, while the last two terms become asymptotically negligible under Assumptions \ref{high dimension assumption}–\ref{sub gaussian assumption}. Specifically, taking \(t = n^{-2r}\) for some \(0 < r < 1/2\), the following result provides a privacy guarantee for our proposed statistics.

\begin{theorem}
\label{privacy guarantee}
Under Assumptions \ref{high dimension assumption}-\ref{sub gaussian assumption}, for any $0 < r < 1/2$, the $\ell_1$ sensitivity of $\Lambda$ satisfies
\begin{equation}
\label{equ: theorem_2}
    \Delta_1 \Lambda \leq \frac{2.01\sigma^2\gamma d}{n}
\end{equation}
with probability at least $1 - 2\exp\{-n^{1-2r}\}$, provided that $n^{-r} \leq c^{\prime}\min \left\{1, \frac{d}{n}\right\}$ for some universal constant $c^{\prime} > 0$. Furthermore, the privatized statistics (\ref{LRTdp})-(\ref{L3dp}) with the noise scale determined by (\ref{equ: theorem_2}) are $\varepsilon$-differentially private with the same probability.
\end{theorem}

\begin{remark}
We can observe that there still exists two unknown parameters $\sigma, \gamma$ 
in the right hand side of (\ref{equ: theorem_2}). In a sense, the sub-Gaussian parameter reflects the variance of the standardized variable $\mathbf{\Sigma}^{-1/2}\mathbf{x}$, e.g., this constant corresponding to a Gaussian variable is $\sigma= 1$. 
For practical considerations, we simply 
assume $\sigma=1$ in this context. 

Recall $K = \min\{d,n\}$ is the number of non-zero eigenvalues of $\mathbf{S_x}$. To estimate the constant \( \gamma \), an intuitive approach is using \( d^{-1} \sum_{i=1}^K \lambda_i \), which is a consistent estimator for \( \mathrm{Tr}(\mathbf{\Sigma})/d \) at the rate of $O_p(n^{-1})$ (see \citep{bai2004clt}). However, privacy concerns complicate this estimation, as adding noise to the sample eigenvalues—necessary to ensure privacy—requires knowing noise scale, which in turn depends on \( \gamma \). To tackle the later issue,  we manually preset a value \( \widetilde{\gamma} \) larger than 1, based on the fact that \(\gamma =1\) under the null hypothesis, which ensures a sufficient amount of noise for privacy protection. Since the noise has zero mean, the law of large numbers ensures that we can still obtain a consistent estimate \( \widehat{\gamma} \) for \( \gamma \), while achieving differential privacy. Further details are provided in Algorithm $\ref{alg:1}$ below.
\end{remark}

\subsection{Privatized statistics and testing procedure} 

To address the issue of ineffective performance associated with individual statistics, we develop a differentially private integrated test statistic that incorporates the three types of statistics (\ref{LRTdp})-(\ref{L3dp}) introduced earlier. Formally, we define
\begin{equation}
    \label{specific functions}
    g_1(x) = |x| - \log |x| - 1,\ g_2(x) = |x-1|^2,\ g_3(x) = |x - 1|,
\end{equation}
 and construct  the integrally privatized statistic as 
\begin{equation}
\label{Tmaxdp}
    T_{max}^{\mathrm{dp}}: = \max_{1\leq m \leq 3}\{T_m^{\mathrm{dp}}\} \quad {\rm{with}} \quad
T_m^{\mathrm{dp}} =  \dfrac{\sqrt{K}|L_{m}^{\mathrm{dp}} - \mu_{0}(g_m)|}{v_{0}^{1/2}(g_m,g_m)} \quad {\rm{for}} \quad
1\leq m \leq 3,
\end{equation}
where $L_{m}^{\mathrm{dp}}, 1\leq m \leq 3$ are defined in (\ref{LRTdp})-(\ref{L3dp}), and $\mu_{0}(g_m), v_{0}(g_m,g_m)/K$ denotes the asymptotic mean and variance of $L_{m}^{\mathrm{dp}}$ for  $ 1\leq m\leq 3$ under $H_0$, respectively. These values can be 
explicitly determined in Theorem~\ref{asymptotic normality 2}.
\begin{remark}{
    In modern data analysis, the structure of underlying signals is typically unknown, and thus the test based on a single, possibly misspecified statistic can undermine detection power. To address this, we construct a robust testing procedure by integrating three complementary loss functions. The entropy-based loss (corresponding to the likelihood ratio) is known to be optimal in classical low-dimensional settings under the Neyman–Pearson framework, and still offer certain advantages in the literature of high-dimensional tests~\citep{bai2009corrections, zheng2019hypothesis}, even though its optimality under high-dimensional or privacy-constrained settings has not been established. To further improve robustness, the squared loss is included for its strong performance under dense alternatives, while the absolute deviation loss contributes  stability when the signal pattern is irregular or difficult to characterize. This integrated approach enhances the adaptability of the test, ensuring more reliable performance across diverse practical scenarios where the structure of the alternative is uncertain.}
\end{remark}
To study the null asymptotics of the privatized statistic, we begin with a general central limit theorem (CLT) for a linear spectral statistics (LSS) of the privatized spectrum $\{\lambda_i+\ell_i\}_{i=1}^K$. Define the  privatized LSS associated with a specific function $g$ as
\begin{equation}\label{general case}
L^{\mathrm{dp}}(g):= \frac{1}{K} \sum_{1\leq i \leq K}g \left(\lambda_i + \ell_i \right).
\end{equation} 
 For a fixed $r\in\NN^+$ and a vector of functions $\mathbf g:=(g_1,\cdots,g_r)^\top$, we first study the asymptotic joint distribution of $L^{\mathrm{dp}}(\mathbf{g}):=\{L^{\mathrm{dp}}(g_1),\cdots,L^{\mathrm{dp}}(g_r)\}^{\top} $, which covers the proposed statistics $(L_1^{\mathrm{dp}},L_2^{\mathrm{dp}},L_3^{\mathrm{dp}})^{\top}$ defined in (\ref{LRTdp})-(\ref{L3dp}) under both the null and the alternative hypotheses as special cases. 
Before going any further, we make some assumptions on the LSD of the sample covariance matrix and certain regularity for functions $\{g_m\}_{m=1}^r$.

\begin{assumption}\label{assumption_Sigma_alternative}
Assume the ESD of the population covariance matrix $\mathbf{\Sigma}$ weakly converges to a nonrandom probability measure $\rho_\kappa$. Let $F_{y,\rho_{\kappa}}$ denote the LSD of the sample covariance matrix $\mathbf{S_x}$. Under the null hypothesis, write $F_{y,\rho_{\kappa}}$ as $F_{y,\rho_{0}}$, corresponding to the M-P law $F_{y}$ defined by (\ref{MP density}), and under the alternatives, write $F_{y,\rho_{\kappa}}$ as $F_{y,\rho_{1}}$, corresponding to the general distribution $F_{y, Q}$ specified in Lemma~\ref{Generalized_MP}.
\end{assumption}

Denote by $f_{y,\rho_{\kappa}}$  the probability density function of  $F_{y,\rho_{\kappa}}$.
Suppose the support of $\rho_{\kappa}$ is $\cup_{k=1}^p [a_{2k},a_{2k-1}]\cap (0,\infty)$, which means that the spectrum bulk of the LSD has possible multiple components, i.e., $p> 1$. Further, we assume regular conditions on every edge $a_k$ and bulk component as in \citep{knowles2017anisotropic}. This regularity is formally specified as follows.  

\begin{assumption}[Regularity of $f_{y,\rho_\kappa}$]\label{regularity}
\begin{enumerate}
\item[(1)] There exists some fixed constant $\tau>0$ such that for each $k=1,\cdots,2p$,  
\$
a_k\ge \tau,\ \min_{1\le k\neq \ell \le p}|a_k-a_\ell|\ge \tau,\ \min_{1\le i\le d}|x_k+\lambda_{i}^{-1}(\mathbf{\Sigma})|\ge \tau,
\$ 
where $\lambda_{i}(\mathbf{\Sigma})$ denotes the $i$-th largest eigenvalue of $\mathbf \Sigma$ and $x_k=m (a_k)$. Here $m(\cdot)$ is defined on $\RR$ by continuous extension from its definition in Lemma \ref{Generalized_MP} on $\CC^+$.

\item[(2)] For any fixed small constant $\tau'>0$, the density $\rho$ in $[a_{2k}+\tau',a_{2k-1}-\tau']$ is bounded from below by some constant $c\equiv c(\tau,\tau')$.  
\end{enumerate}
\end{assumption}

\begin{assumption}[Lyapunov's condition]
    \label{Lyapunov assumption}
    For any nonrandom sequence  $\{\theta_i\}_{i = 1}^K \subseteq [a,b]$ with $0\leq a \leq b$, assume $\{g_m(\theta_i+\ell_i)\}_{i=1}^K$ satisfies the Lyapunov's condition for all $1\le m\le r$, that is, there exists some $\delta>0$ such that
\begin{equation}\label{lyapunov}
    \lim_{K \to \infty} \frac{1}{B_{K,m}^{2+\delta}} \sum_{1\leq i \leq K} \mathbb E\left\{ \left|g_m(\theta_i + \ell_i) - \mathbb Eg_m(\theta_i + \ell_i) \right|^{2+\delta}\right\} = 0,\quad \forall 1\leq m\leq r,
\end{equation}
where $B_{K,m}^2 := \sum_{1\leq i \leq K}\mathrm{Var}\{g_m(\theta_i + \ell_i)\}$.
\end{assumption}

\begin{assumption}[Regularity of $\bf g$]
\label{assumption approximation error}
If a bounded nonrandom sequence $\{\eta_i\}_{i=1}^K$ can well approximate an almost surely bounded random sequence $\{\xi_i\}_{i =1 }^K$ such that there exists some constant $c>0$, 

\#\label{approximation_condition}
\PP\left(|\xi_i-\eta_i|\le c K^{-1/2-\epsilon},\ \forall 1\le i\le K\right)=1-o(1),\quad K \to \infty,
\#
for any small constant $\epsilon>0$, 
\begin{equation}\label{approximately equal}
\frac{1}{K} \sum_{1\le i\le K} g_m(\xi_{i}+\ell_i)
= \frac{1}{K} \sum_{1\le i\le K} g_m(\eta_i+\ell_i) + o_p(K^{-1/2}), \quad 1\leq m\leq r, 
\end{equation}
where $\{\ell_i\}_{i=1}^K$ represents an i.i.d. Laplacian noise with zero mean and a constant variance.
\end{assumption}

Define the mean vector $\mu_{\kappa}(\mathbf{g}):= \{\mu_{\kappa}(g_1), \cdots, \mu_{\kappa}(g_r)\}^{\top}$  with
\begin{equation}
     \label{mean vector}
    \mu_{\kappa}(g_m) = (1\vee y) \int_{t>0}b_{g_m}(t)\mathrm{d}F_{y,\rho_{\kappa}}(t),\ 1\leq m \leq r,
\end{equation}
and the covariance matrix $V_{\kappa}(\mathbf{g}):=[v_{\kappa}(g_m,g_s)]_{m,s = 1}^r$ with
\begin{equation}
    \label{covariance}
    v_{\kappa}(g_m,g_s) = (1\vee y)\int_{t>0} \left\{b_{g_{m} ,g_{s}}(t) - b_{g_{m}}(t) b_{g_{s}}(t)\right\}{\rm d} F_{y,\rho_{\kappa}} (t),
\end{equation}
where $b_{g_m}(t)$ and $b_{g_{m}, g_{s}}(t), 1\leq m,s \leq r$ are given by
\begin{equation*}
    \label{median functions}
    b_{g_m}(t) = \int g_{m}(x+t)\,\mathrm{d}H(x) \quad {\rm{and}} \quad  b_{g_{m},g_{s}}(t) = \int g_{m}(x+t)g_{s}(x+t)\,\mathrm{d}H(x),
\end{equation*}
respectively. Here $H(x)$ denotes the probability distribution function of the Laplacian noise, 
and $F_{y,\rho_{\kappa}}$ denotes the LSD of the sample covariance matrix specified in Assumption~\ref{assumption_Sigma_alternative}, where ``$\,\kappa$" takes $0$ or $1$, corresponding to the null hypothesis and the alternative hypothesis, respectively.   
If the covariance matrix defined by \eqref{covariance} is non-degenerated, the following lemma provides the asymptotic joint normality of $L^{\mathrm{dp}}(\mathbf{g})$, including the null and alternative cases. 

\begin{lemma}[Asymptotic normality]\label{asymptotic normality}
Under Assumptions \ref{high dimension assumption}-\ref{assumption approximation error}, it holds that as $d\to \infty$, 
\$
\sqrt{K}\left\{L^{\mathrm{dp}}(\mathbf{g}) - \mu_{\kappa}(\mathbf{g})\right\}\xrightarrow{D} \mathcal{N}_r(\mathbf{0}, V_{\kappa}(\mathbf{g})), 
\$
where the mean vector $\mu_{\kappa}(\mathbf{g})$ and the covariance matrix $V_{\kappa}(\mathbf{g})$ are defined by \eqref{mean vector} and  \eqref{covariance}, respectively.
\end{lemma}

\begin{remark}
This lemma can be viewed as a DP version of the LSS results obtained in \citep{bai2004clt}. 
It is proved using tools from RMT. Mathematically, most nonzero eigenvalues of $\mathbf{S_x}$, that is, those within the bulk of the spectrum, are located around their ``classical location'' and fluctuate at a magnitude of $O(K^{-1})$, which is called rigidity. This implies a convergence rate $O_p(K^{-1})$ for LSS $\frac{1}{K}\sum_{1\le i\le K}g(\lambda_i)$ for a specific function $g$. However, the Laplacian noise terms introduced $\{\ell_i\}_{i=1}^K$ in \eqref{general case} lead to a fluctuation of $O_p(K^{-1/2})$ after averaging and scaling by $\sqrt{K}$, thus dominating the asymptotic variance of the perturbed LSS, specifically the privateized statistic $L^{\mathrm{dp}}(g)$. This suggests that the privatized statistic exhibits a different magnitude of asymptotic variance compared to the original LSS. Moreover, from a privacy perspective, it is impossible to recover $\{\lambda_i\}_{i=1}^K$ from $\{\lambda_i+\ell_i\}_{i=1}^K$ with LSS consistency. This lemma, along with its technical contributions, paves the way for a systematic understanding of differential privacy in high-dimensional statistics. 
\end{remark}

Under the null hypothesis $H_0$, Assumption~\ref{assumption_Sigma_alternative} holds automatically. By checking the Lyapunov's condition specified in Assumption~\ref{Lyapunov assumption} and the regularity in Assumption~\ref{assumption approximation error} for functions defined in (\ref{specific functions}), we obtain the following result.
\begin{theorem}[Asymptotic null distribution]\label{asymptotic normality 2}
Under Assumptions \ref{high dimension assumption}-\ref{sub gaussian assumption}, for the specific functions defined in~(\ref{specific functions}), it holds that as $d\to \infty$, 
$$
\sqrt{K}\left\{L^{\mathrm{dp}}(\mathbf{g}) - \mu_{0}(\mathbf{g})\right\} \xrightarrow{D} \mathcal{N}_3(\mathbf{0}, V_{0}(\mathbf{g})),
$$
 where the mean vector $\mu_{0}(\mathbf{g})$ and the covariance matrix $V_0(\mathbf{g})$ are determined in Lemma~\ref{asymptotic normality} by taking $F_{y,\rho_{\kappa}}$ as $F_{y,\rho_0}$.
\end{theorem} 

Theorem~\ref{asymptotic normality 2} indicates that the proposed test is asymptotically distribution-free, and we reject $H_0$ if and only if $|T_{max}^{\mathrm{dp}}|>z_{\alpha}$,  where $z_{\alpha}$ denotes the $(1-\alpha)$ quantile of $\max_{1\leq m\leq 3} |Y_m|$, where $(Y_1,Y_2,Y_3)^{\top}\sim \mathcal{N}_3(\mathrm{0},\mathbf{\Gamma}_0^{-1}V_0(\mathbf{g})\mathbf{\Gamma}_0^{-1})$ with $\mathbf{\Gamma}_0= \mathrm{diag}\{v_0^{1/2}(g_1,g_1), v_0^{1/2}(g_2,g_2),v_0^{1/2}(g_3,g_3)\}$. As demonstrated in simulation studies, the empirical Type I errors via its null limiting distribution can be substantially controlled without additional Monte Carlo approximation or resampling techniques for critical value determination. We summarize the procedure of the privatized test for problem (\ref{simple hypothesis}) in Algorithm \ref{alg:1}.

\begin{algorithm}[H]
  	\setstretch{2}
  	\caption{Privatized testing procedure}
  	\label{alg:1}
  	\renewcommand{\algorithmicrequire}{\textbf{Input:}}
  	\renewcommand{\algorithmicensure}{\textbf{Output:}}
  	\begin{algorithmic}[1]
  		\REQUIRE The dataset $\mathbf{X} \in \mathbb R^{n \times d }$, privacy parameter $\varepsilon$, preset parameter $\widetilde{\gamma}$,  mean vector $\mu_0(\mathbf{g})$, covariance matrix $V_0(\mathbf{g})$, the significance level $\alpha$.
  	\STATE Calculate non-zero eigenvalues of $\mathbf{S_x}= n^{-1}\mathbf{X}^{\top}\mathbf{X}$, denoted by $({\lambda}_1,\ldots, \lambda_K)^{\top}$. 
  		\STATE Generate Laplacian noise $\ell_i^{\ast} \sim \text{Lap}\left(0,\dfrac{2.01\widetilde{\gamma}d}{n\varepsilon}\right)$, and set $\lambda_i^{\ast} = {\lambda}_i + \ell_i^{\ast}$ for $1\leq i \leq K$.
    \STATE Estimate $\gamma$ by $\widehat{\gamma}=d^{-1}\cdot {\left|\sum_{i=1}^K\lambda_i^{\ast}\right|}$.
    \STATE Generate Laplacian noise $\ell_i \sim \text{Lap}\left(0,\dfrac{2.01\widehat{\gamma}d}{n\varepsilon}\right)$, and set $\tilde{\lambda}_i = {\lambda}_i + \ell_i$ for $1\leq i \leq K$.
  		\STATE Construct privatized  test statistics defined in (\ref{LRTdp}), (\ref{L2dp}) and (\ref{L3dp}), respectively.
    \STATE  Integrate the statistics (\ref{LRTdp}), (\ref{L2dp}) and (\ref{L3dp}) to get $T_{max}^{\mathrm{dp}}$ as defined in (\ref{Tmaxdp}).
 \ENSURE Reject (or accept) $H_0$ if and only if $T_{max}^{\mathrm{dp}} > z_{\alpha}$ ($T_{max}^{\mathrm{dp}} \leq  z_{\alpha}$, respectively).
  		
  	\end{algorithmic}
  \end{algorithm}

   \begin{remark}
In the conventional framework of differential privacy (DP) algorithm design, a trusted data curator typically assumes responsibility for aggregating raw datasets and disseminating privacy-preserving outputs to analysts through carefully designed mechanisms~\citep{dwork2014algorithmic}. 

In Algorithm~\ref{alg:1}, the curator first computes the sample eigenvalues $\{\lambda_i\}_{i = 1}^K$ based on the collected data, perturbs each eigenvalue using the Laplace mechanism, and then releases the privatized eigenvalues $\{\tilde\lambda_i\}_{i = 1}^K$. These released eigenvalues are subsequently used for various hypothesis testing tasks. Since no information regarding the eigenvectors is disclosed, the analyst cannot reconstruct the sample covariance matrix $\mathbf{S_x}$, thereby providing a strong level of privacy protection. In addition, the curator can also release the scale of the added noise, which depends on the parameter $\gamma$. As mentioned earlier, a straightforward forward estimate for $\gamma$ is $d^{-1}\sum_{i = 1}^K\lambda_i$. This estimation is data-dependent and potentially sensitive from a privacy perspective, and thus must also be performed in a privacy-preserving manner. This consideration is explicitly addressed in our algorithm by introducing a perturbed estimate $\widehat{\gamma}$.
   \end{remark}

\subsection{Power study}

In this subsection, we discuss the asymptotic properties of $T_{max}^{\mathrm{dp}}$ under alternative hypotheses.

In general cases, the sample covariance matrix is $\mathbf{S_x}= n^{-1}\mathbf{\Sigma}^{1/2} \mathbf Z^\top \mathbf{Z} \mathbf{\Sigma}^{1/2}$, where rows of $\mathbf Z$ are sub-Gaussian as specified in Assumption~\ref{sub gaussian assumption}. 
Similar to Theorem~\ref{asymptotic normality 2}, we obtain the asymptotic joint distribution of $(L_1^{\mathrm{dp}},L_2^{\mathrm{dp}},L_3^{\mathrm{dp}})^{\top}$ under alternative hypotheses.

\begin{theorem}[Asymptotic alternative distribution]\label{CLT_H1}
Under Assumptions \ref{high dimension assumption}-\ref{regularity},  for the specific functions defined in~(\ref{specific functions}), it holds that  as $d\to \infty$, 
$$
\sqrt{K}\left\{L^{\mathrm{dp}}(\mathbf{g}) - \mu_{1}(\mathbf{g})\right\} \xrightarrow{D} \mathcal{N}_3(\mathbf{0}, V_{1}(\mathbf{g})),
$$
where the mean vector $\mu_{1}(\mathbf{g})$ and the covariance matrix $V_1(\mathbf{g})$ are determined in Lemma~\ref{asymptotic normality} by taking $F_{y,\rho_{\kappa}}$ as $F_{y,\rho_1}$.
\end{theorem} 

Theorems~\ref{asymptotic normality 2} and~\ref{CLT_H1} indicate that the asymptotical variances of $L_{m}^{\mathrm{dp}}$ for $ 1\leq m \leq 3$ are $O(K^{-1})$ under both $H_0$ and $H_1$. Denote the power function of  $T_{max}^{\mathrm{dp}} $ by $\mathrm{Power}(T_{max}^{\mathrm{dp}})$ and write $\mu_{0,m} := \mu_0(g_m), \mu_{1,m} := \mu_1(g_m),s_{0,m}:=v_0^{1/2}(g_m,g_m),s_{1,m}:=v_1^{1/2}(g_m,g_m)$ for short. Then 
$$ 
\begin{aligned}&\,\,\quad 1-\mathrm{Power}(T_{max}^{\mathrm{dp}})\\&= \PP \left(T_{max}^{\mathrm{dp}}\leq z_{\alpha} \mid H_1\ {\rm is\ true}\right) \\
& = \mathbb{P} \left(\max_{1\leq m\leq 3}\left|\dfrac{\sqrt{K}(L_{m}^{\mathrm{dp}} - \mu_{0,m})}{s_{0,m}} \right|\leq z_{\alpha}\,\,\Big|\,\, H_1\ {\rm is\ true}\right)\\
&= \PP \left\{ \max_{1\leq m\leq 3}\left|\frac{\sqrt K (L_{m}^{\mathrm{dp}}-\mu_{1,m})}{s_{1,m}}- \frac{\sqrt K}{s_{1,m}}(\mu_{0,m}-\mu_{1,m})\right| \leq  \frac{s_{0,m}}{s_{1,m}}z_{\alpha}\,\,\Big|\,\, H_1\ {\rm is\ true}\right\}.
\end{aligned}
$$
By the asymptotical normality of $L_{m}^{\mathrm{dp}}, 1\leq m \leq 3$  established in Theorem~\ref{CLT_H1}, the last display can be approximated by  
$$
\PP \left( \max_{1\leq m\leq 3}\left|Y_m - A_{m}\right| \leq B_{m}z_{\alpha}\right)\ \text{with}\ (Y_1,Y_2,Y_3)^{\top} \sim \mathcal{N}_3(\mathbf{0},\mathbf{\Gamma}_1^{-1}V_1(\mathbf{g})\mathbf{\Gamma}_1^{-1}),
$$
where $\mathbf{\Gamma}_1 = \mathrm{diag}\{s_{1,1},s_{1,2},s_{1,3}\}$. If there exists some $1\leq m \leq 3$ such that $\sqrt{K}|\mu_{0,m} -\mu_{1,m} | \to \infty$, then $|A_{m}| = \infty$ and $B_{m}$ is of constant order.  Consequently, we have  $\mathrm{Power}(T_{max}^{\mathrm{dp}}) \rightarrow 1$ as $n \to \infty$. This implies that our privatized test statistic can consistently detect alternative hypotheses if and only if $\max_{1\leq m\leq 3} \sqrt{K}|\mu_{0,m} -\mu_{1,m} | \to \infty$.
 Furthermore, if $\max_{1\leq m \leq 3}\sqrt{K}|\mu_{0,m} -\mu_{1,m} | \to c$ for some finite constant $c>0$, then there exists some $1\leq m \leq 3$ such that $\sqrt{K}|\mu_{0,m} -\mu_{1,m} | \to c$ given the convergence of statistics, and thus $|A_m| >0$. Together with the fact that $B_k =1$ for all $1\leq k\leq 3$, this leads to $\mathrm{Power}(T_{max}^{\mathrm{dp}}) > \alpha$ as $n\to\infty$. Therefore, we conclude that the privatized test statistic detects local alternative hypotheses distinct from the null at a rate of $1/\sqrt{n}$.

\section{Numerical studies}
	\label{section:4}
\subsection{Simulation}
In this subsection, we conduct simulations to illustrate the performance of the proposed test. In the following simulations, the sample size $n$ takes values from $\{400, 600, 800\}$, the ratio $y$ of dimension to sample size takes values from $\{0.5, 1, 5\}$, the privacy parameter $\varepsilon$ takes values from $\{1, 2, 4, 8\}$ and the preset parameter $\widetilde{\gamma}$ is chosen to be $2$.  We take the significance level $\alpha = 0.05$. All experiments are repeated $2,000$ times to compute the empirical sizes and powers. We consider two types of populations: Gaussian and Uniform distributions. Specifically, we generate data from the following two models:
\begin{itemize}
\setlength{\itemindent}{0pt}
\item \textrm{Model I:} $\mathbf{x} = \mathbf{\Sigma}^{1/2}\mathbf{z}$,  with $\mathbf{z}\sim \mathcal{N}(\mathbf{0},\mathbf{I}_d)$,
\item \textrm{Model II:} $\mathbf{x} = \mathbf{\Sigma}^{1/2}\mathbf{z}$, with 
 $\mathbf{z}\sim \textrm{Unif}\left([\sqrt{3},\sqrt{3}]^d\right)$.
\end{itemize}
We first consider the covariance matrix with the following structure:
\begin{align*}
\mathbf{\Sigma} = (1+ \delta)\cdot \mathbf{I}_d,
\end{align*}
where the parameter $\delta$ takes values from $\{0,\pm 0.25,\pm0.5 \}$. In particular, the value $\delta = 0$ corresponds to the null hypothesis and $\delta \neq 0$ to the alternative hypotheses. The empirical sizes and powers of the privatized test are reported in Tables \ref{table_model_I}-\ref{table_model_II}.
\begin{landscape}
\begin{table}[htb!]
\tiny
\centering
\setlength
\tabcolsep{6pt}
\renewcommand{\arraystretch}{1.35}
\caption{Empirical sizes and powers for Model I}
\label{table_model_I}
\begin{tabular}{c|c|cccc|cccc|cccc|cccc|cccc}
\hline
&  & \multicolumn{4}{c|}{$\delta = 0$} & \multicolumn{4}{c|}{$\delta = -0.25$}  & \multicolumn{4}{c|}{$\delta = -0.5$} & \multicolumn{4}{c|}{$\delta = 0.25$} & \multicolumn{4}{c}{$\delta = 0.5$} \\ 
\hline
& $(n,d)$   & $\varepsilon = 1$ & $\varepsilon = 2$ & $\varepsilon = 4$ & $\varepsilon = 8$ & $\varepsilon = 1$ & $\varepsilon = 2$ & $\varepsilon = 4$ & $\varepsilon = 8$ & $\varepsilon = 1$ & $\varepsilon = 2$ & $\varepsilon = 4$ & $\varepsilon = 8$ & $\varepsilon = 1$ & $\varepsilon = 2$ & $\varepsilon = 4$ & $\varepsilon = 8$ & $\varepsilon = 1$ & $\varepsilon = 2$ & $\varepsilon = 4$ & $\varepsilon = 8$ \\ \hline
\multirow{3}{*}{$T_1^{\mathrm{dp}}$}     & (400,200) & 0.043             & 0.053             & 0.053             & 0.048             & 0.063             & 0.102             & 0.260             & 0.501             & 0.447             & 0.930             & 1                 & 1                 & 0.179             & 0.171             & 0.108             & 0.070             & 0.484             & 0.799             & 0.723             & 0.743             \\
& (600,300) & 0.046   & 0.058   & 0.048   & 0.059   & 0.079  & 0.113    & 0.344  & 0.658             & 0.518             & 0.984             & 1                 & 1                 & 0.221             & 0.237             & 0.132             & 0.093             & 0.655             & 0.926             & 0.904             & 0.922             \\
                           & (800,400) & 0.054             & 0.055             & 0.049             & 0.043             & 0.085             & 0.115             & 0.382             & 0.778             & 0.596             & 0.996             & 1                 & 1                 & 0.294             & 0.303             & 0.174             & 0.137             & 0.767             & 0.979             & 0.967             & 0.984             \\ \hline
\multirow{3}{*}{$T_2^{\mathrm{dp}}$}     & (400,200) & 0.054             & 0.049             & 0.059             & 0.076             & 0.149             & 0.635             & 1                 & 1                 & 0.346             & 0.859             & 1                 & 1                 & 0.131             & 0.602             & 1                 & 1                 & 0.228             & 0.969             & 1                 & 1                 \\
                           & (600,300) & 0.052             & 0.052             & 0.045             & 0.065             & 0.188             & 0.815             & 1                 & 1                 & 0.390             & 0.945             & 1                 & 1                 & 0.128             & 0.769             & 1                 & 1                 & 0.316             & 0.998             & 1                 & 1                 \\
                           & (800,400) & 0.045             & 0.049             & 0.055             & 0.054             & 0.225             & 0.893             & 1                 & 1                 & 0.467             & 0.974             & 1                 & 1                 & 0.175             & 0.870             & 1                 & 1                 & 0.372             & 1                 & 1                 & 1                 \\ \hline
\multirow{3}{*}{$T_3^{\mathrm{dp}}$}     & (400,200) & 0.058             & 0.058             & 0.060             & 0.062             & 0.209             & 0.635             & 0.994             & 1                 & 0.313             & 0.659             & 0.984             & 1                 & 0.167             & 0.628             & 0.999             & 1                 & 0.370             & 0.985             & 1                 & 1                 \\
                           & (600,300) & 0.050             & 0.052             & 0.051             & 0.054             & 0.256             & 0.806             & 1                 & 1                 & 0.371             & 0.802             & 0.998             & 1                 & 0.197             & 0.799             & 1                 & 1                 & 0.527             & 1.000             & 1                 & 1                 \\
                           & (800,400) & 0.045             & 0.049             & 0.052             & 0.051             & 0.326             & 0.882             & 1                 & 1                 & 0.457             & 0.905             & 1                 & 1                 & 0.262             & 0.904             & 1                 & 1                 & 0.610             & 1                 & 1                 & 1                 \\ \hline
\multirow{3}{*}{$T_{max}^{\mathrm{dp}}$} & (400,200) & 0.054             & 0.057             & 0.064             & 0.076             & 0.156             & 0.600             & 1                 & 1                 & 0.482             & 0.968             & 1                 & 1                 & 0.177             & 0.578             & 1                 & 1                 & 0.429             & 0.979             & 1                 & 1                 \\
                           & (600,300) & 0.046             & 0.061             & 0.044             & 0.060             & 0.205             & 0.775             & 1                 & 1                 & 0.563             & 0.997             & 1                 & 1                 & 0.208             & 0.757             & 1                 & 1                 & 0.609             & 0.998             & 1                 & 1                 \\
                           & (800,400) & 0.053             & 0.052             & 0.053             & 0.053             & 0.258             & 0.868             & 1                 & 1                 & 0.657             & 0.999             & 1                 & 1                 & 0.286             & 0.865             & 1                 & 1                 & 0.717             & 1                 & 1                 & 1                 \\ \hline
\end{tabular}
\vspace{1.5em}

\begin{tabular}{c|c|cccc|cccc|cccc|cccc|cccc}
\hline
                           &           & \multicolumn{4}{c|}{$\delta = 0$}                                             & \multicolumn{4}{c|}{$\delta = -0.25$}                                         & \multicolumn{4}{c|}{$\delta = -0.5$}                                          & \multicolumn{4}{c|}{$\delta = 0.25$}                                          & \multicolumn{4}{c}{$\delta = 0.5$}                                            \\ \hline
                           & $(n,d)$   & $\varepsilon = 1$ & $\varepsilon = 2$ & $\varepsilon = 4$ & $\varepsilon = 8$ & $\varepsilon = 1$ & $\varepsilon = 2$ & $\varepsilon = 4$ & $\varepsilon = 8$ & $\varepsilon = 1$ & $\varepsilon = 2$ & $\varepsilon = 4$ & $\varepsilon = 8$ & $\varepsilon = 1$ & $\varepsilon = 2$ & $\varepsilon = 4$ & $\varepsilon = 8$ & $\varepsilon = 1$ & $\varepsilon = 2$ & $\varepsilon = 4$ & $\varepsilon = 8$ \\ \hline
\multirow{3}{*}{$T_1^{\mathrm{dp}}$}     & (400,400) & 0.051             & 0.053             & 0.054             & 0.048             & 0.219             & 0.326             & 0.189             & 0.097             & 0.453             & 0.214             & 0.171             & 0.600             & 0.188             & 0.515             & 0.695             & 0.577             & 0.339             & 0.966             & 1                 & 1                 \\
                           & (600,600) & 0.058             & 0.046             & 0.044             & 0.057             & 0.325             & 0.452             & 0.256             & 0.110             & 0.595             & 0.260             & 0.188             & 0.752             & 0.210             & 0.685             & 0.855             & 0.782             & 0.485             & 0.997             & 1                 & 1                 \\
                           & (800,800) & 0.056             & 0.053             & 0.047             & 0.044             & 0.423             & 0.556             & 0.313             & 0.143             & 0.714             & 0.282             & 0.229             & 0.820             & 0.279             & 0.805             & 0.931             & 0.882             & 0.556             & 0.999             & 1                 & 1                 \\ \hline
\multirow{3}{*}{$T_2^{\mathrm{dp}}$}     & (400,400) & 0.046             & 0.050             & 0.050             & 0.059             & 0.212             & 0.646             & 1                 & 1                 & 0.562             & 0.984             & 1                 & 1                 & 0.089             & 0.363             & 0.997             & 1                 & 0.143             & 0.740             & 1                 & 1                 \\
                           & (600,600) & 0.052             & 0.044             & 0.049             & 0.053             & 0.209             & 0.795             & 1                 & 1                 & 0.623             & 0.999             & 1                 & 1                 & 0.101             & 0.504             & 1                 & 1                 & 0.188             & 0.886             & 1                 & 1                 \\
                           & (800,800) & 0.051             & 0.045             & 0.053             & 0.054             & 0.238             & 0.877             & 1                 & 1                 & 0.708             & 1                 & 1                 & 1                 & 0.114             & 0.596             & 1                 & 1                 & 0.199             & 0.950             & 1                 & 1                 \\ \hline
\multirow{3}{*}{$T_3^{\mathrm{dp}}$}     & (400,400) & 0.043             & 0.045             & 0.047             & 0.057             & 0.274             & 0.724             & 1                 & 1                 & 0.680             & 0.993             & 1                 & 1                 & 0.126             & 0.466             & 0.991             & 1                 & 0.228             & 0.881             & 1                 & 1                 \\
                           & (600,600) & 0.052             & 0.044             & 0.046             & 0.055             & 0.302             & 0.873             & 1                 & 1                 & 0.753             & 0.999             & 1                 & 1                 & 0.138             & 0.638             & 0.999             & 1                 & 0.332             & 0.974             & 1                 & 1                 \\
                           & (800,800) & 0.051             & 0.050             & 0.047             & 0.061             & 0.364             & 0.938             & 1                 & 1                 & 0.835             & 1                 & 1                 & 1                 & 0.189             & 0.722             & 1                 & 1                 & 0.375             & 0.990             & 1                 & 1                 \\ \hline
\multirow{3}{*}{$T_{max}^{\mathrm{dp}}$} & (400,400) & 0.047             & 0.056             & 0.052             & 0.055             & 0.277             & 0.701             & 1                 & 1                 & 0.724             & 0.991             & 1                 & 1                 & 0.162             & 0.522             & 0.995             & 1                 & 0.302             & 0.944             & 1                 & 1                 \\
                           & (600,600) & 0.057             & 0.043             & 0.048             & 0.059             & 0.348             & 0.858             & 1                 & 1                 & 0.816             & 0.999             & 1                 & 1                 & 0.174             & 0.690             & 0.999             & 1                 & 0.417             & 0.991             & 1                 & 1                 \\
                           & (800,800) & 0.052             & 0.054             & 0.043             & 0.056             & 0.432             & 0.919             & 1                 & 1                 & 0.896             & 1                 & 1                 & 1                 & 0.230             & 0.797             & 1                 & 1                 & 0.478             & 0.998             & 1                 & 1                 \\ \hline
\end{tabular}

\vspace{1.5em}

\begin{tabular}{c|c|cccc|cccc|cccc|cccc|cccc}
\hline
                           &            & \multicolumn{4}{c|}{$\delta = 0$}                                             & \multicolumn{4}{c|}{$\delta = -0.25$}                                         & \multicolumn{4}{c|}{$\delta = -0.5$}                                          & \multicolumn{4}{c|}{$\delta = 0.25$}                                          & \multicolumn{4}{c}{$\delta = 0.5$}                                            \\ \hline
                           & $(n,d)$    & $\varepsilon = 1$ & $\varepsilon = 2$ & $\varepsilon = 4$ & $\varepsilon = 8$ & $\varepsilon = 1$ & $\varepsilon = 2$ & $\varepsilon = 4$ & $\varepsilon = 8$ & $\varepsilon = 1$ & $\varepsilon = 2$ & $\varepsilon = 4$ & $\varepsilon = 8$ & $\varepsilon = 1$ & $\varepsilon = 2$ & $\varepsilon = 4$ & $\varepsilon = 8$ & $\varepsilon = 1$ & $\varepsilon = 2$ & $\varepsilon = 4$ & $\varepsilon = 8$ \\ \hline
\multirow{3}{*}{$T_1^{\mathrm{dp}}$}     & (400,2000) & 0.050             & 0.057             & 0.046             & 0.048             & 0.421             & 0.951             & 1                 & 1                 & 0.880             & 1                 & 1                 & 1                 & 0.153             & 0.613             & 1                 & 1                 & 0.260             & 0.941             & 1                 & 1                 \\
                           & (600,3000) & 0.050             & 0.043             & 0.053             & 0.047             & 0.493             & 0.994             & 1                 & 1                 & 0.946             & 1                 & 1                 & 1                 & 0.185             & 0.782             & 1                 & 1                 & 0.353             & 0.989             & 1                 & 1                 \\
                           & (800,4000) & 0.047             & 0.037             & 0.052             & 0.048             & 0.578             & 0.999             & 1                 & 1                 & 0.980             & 1                 & 1                 & 1                 & 0.222             & 0.881             & 1                 & 1                 & 0.437             & 1                 & 1                 & 1                 \\ \hline
\multirow{3}{*}{$T_2^{\mathrm{dp}}$}     & (400,2000) & 0.046             & 0.060             & 0.045             & 0.050             & 0.218             & 0.775             & 1                 & 1                 & 0.681             & 1                 & 1                 & 1                 & 0.083             & 0.311             & 0.997             & 1                 & 0.122             & 0.648             & 1                 & 1                 \\
                           & (600,3000) & 0.050             & 0.041             & 0.058             & 0.044             & 0.265             & 0.887             & 1                 & 1                 & 0.759             & 1                 & 1                 & 1                 & 0.097             & 0.455             & 1                 & 1                 & 0.145             & 0.801             & 1                 & 1                 \\
                           & (800,4000) & 0.051             & 0.047             & 0.052             & 0.051             & 0.278             & 0.953             & 1                 & 1                 & 0.845             & 1                 & 1                 & 1                 & 0.111             & 0.557             & 1                 & 1                 & 0.166             & 0.887             & 1                 & 1                 \\ \hline
\multirow{3}{*}{$T_3^{\mathrm{dp}}$}     & (400,2000) & 0.050             & 0.049             & 0.046             & 0.044             & 0.388             & 0.941             & 1                 & 1                 & 0.848             & 1                 & 1                 & 1                 & 0.138             & 0.572             & 1                 & 1                 & 0.240             & 0.920             & 1                 & 1                 \\
                           & (600,3000) & 0.050             & 0.044             & 0.057             & 0.054             & 0.449             & 0.990             & 1                 & 1                 & 0.930             & 1                 & 1                 & 1                 & 0.170             & 0.741             & 1                 & 1                 & 0.316             & 0.984             & 1                 & 1                 \\
                           & (800,4000) & 0.046             & 0.039             & 0.052             & 0.050             & 0.530             & 0.998             & 1                 & 1                 & 0.963             & 1                 & 1                 & 1                 & 0.198             & 0.845             & 1                 & 1                 & 0.392             & 0.998             & 1                 & 1                 \\ \hline
\multirow{3}{*}{$T_{max}^{\mathrm{dp}}$} & (400,2000) & 0.043             & 0.058             & 0.044             & 0.048             & 0.381             & 0.940             & 1                 & 1                 & 0.861             & 1                 & 1                 & 1                 & 0.133             & 0.569             & 1                 & 1                 & 0.235             & 0.922             & 1                 & 1                 \\
                           & (600,3000) & 0.049             & 0.041             & 0.052             & 0.049             & 0.453             & 0.990             & 1                 & 1                 & 0.937             & 1                 & 1                 & 1                 & 0.160             & 0.732             & 1                 & 1                 & 0.302             & 0.986             & 1                 & 1                 \\
                           & (800,4000) & 0.047             & 0.045             & 0.046             & 0.048             & 0.530             & 0.999             & 1                 & 1                 & 0.971             & 1                 & 1                 & 1                 & 0.190             & 0.850             & 1                 & 1                 & 0.385             & 0.999             & 1                 & 1                 \\ \hline
\end{tabular}
\end{table}
\end{landscape}

\begin{landscape}
\begin{table}[htb!]
\tiny
		 	\centering
            \setlength
		 	\tabcolsep{6pt}
		 	\renewcommand{\arraystretch}{1.35}
            \caption{Empirical sizes and powers for Model II}
            \label{table_model_II}
\begin{tabular}{c|c|cccc|cccc|cccc|cccc|cccc}
\hline
                           &           & \multicolumn{4}{c|}{$\delta = 0$}                                             & \multicolumn{4}{c|}{$\delta = -0.25$}                                         & \multicolumn{4}{c|}{$\delta = -0.5$}                                          & \multicolumn{4}{c|}{$\delta = 0.25$}                                          & \multicolumn{4}{c}{$\delta = 0.5$}                                            \\ \hline
                           & $(n,d)$   & $\varepsilon = 1$ & $\varepsilon = 2$ & $\varepsilon = 4$ & $\varepsilon = 8$ & $\varepsilon = 1$ & $\varepsilon = 2$ & $\varepsilon = 4$ & $\varepsilon = 8$ & $\varepsilon = 1$ & $\varepsilon = 2$ & $\varepsilon = 4$ & $\varepsilon = 8$ & $\varepsilon = 1$ & $\varepsilon = 2$ & $\varepsilon = 4$ & $\varepsilon = 8$ & $\varepsilon = 1$ & $\varepsilon = 2$ & $\varepsilon = 4$ & $\varepsilon = 8$ \\ \hline
\multirow{3}{*}{$T_1^{\mathrm{dp}}$}     & (400,200) & 0.055             & 0.048             & 0.048             & 0.048             & 0.074             & 0.105             & 0.224             & 0.477             & 0.453             & 0.935             & 1                 & 1                 & 0.171             & 0.151             & 0.094             & 0.079             & 0.487             & 0.763             & 0.726             & 0.711             \\
                           & (600,300) & 0.054             & 0.059             & 0.048             & 0.047             & 0.072             & 0.107             & 0.306             & 0.643             & 0.519             & 0.983             & 1                 & 1                 & 0.237             & 0.218             & 0.134             & 0.100             & 0.662             & 0.918             & 0.896             & 0.920             \\
                           & (800,400) & 0.041             & 0.049             & 0.048             & 0.053             & 0.078             & 0.123             & 0.376             & 0.756             & 0.583             & 0.995             & 1                 & 1                 & 0.275             & 0.266             & 0.158             & 0.124             & 0.778             & 0.974             & 0.966             & 0.980             \\ \hline
\multirow{3}{*}{$T_2^{\mathrm{dp}}$}     & (400,200) & 0.051             & 0.055             & 0.044             & 0.066             & 0.164             & 0.643             & 1                 & 1                 & 0.331             & 0.843             & 1                 & 1                 & 0.135             & 0.578             & 1                 & 1                 & 0.250             & 0.975             & 1                 & 1                 \\
                           & (600,300) & 0.048             & 0.043             & 0.047             & 0.060             & 0.192             & 0.807             & 1                 & 1                 & 0.410             & 0.952             & 1                 & 1                 & 0.145             & 0.751             & 1                 & 1                 & 0.301             & 0.998             & 1                 & 1                 \\
                           & (800,400) & 0.041             & 0.047             & 0.055             & 0.054             & 0.210             & 0.902             & 1                 & 1                 & 0.469             & 0.980             & 1                 & 1                 & 0.166             & 0.859             & 1                 & 1                 & 0.375             & 1                 & 1                 & 1                 \\ \hline
\multirow{3}{*}{$T_3^{\mathrm{dp}}$}     & (400,200) & 0.057             & 0.055             & 0.044             & 0.055             & 0.217             & 0.638             & 0.996             & 1                 & 0.308             & 0.663             & 0.990             & 1                 & 0.158             & 0.616             & 0.996             & 1                 & 0.388             & 0.989             & 1                 & 1                 \\
                           & (600,300) & 0.055             & 0.048             & 0.045             & 0.050             & 0.269             & 0.801             & 1                 & 1                 & 0.381             & 0.823             & 0.999             & 1                 & 0.208             & 0.792             & 1                 & 1                 & 0.522             & 1                 & 1                 & 1                 \\
                           & (800,400) & 0.043             & 0.054             & 0.057             & 0.047             & 0.318             & 0.902             & 1                 & 1                 & 0.457             & 0.897             & 1                 & 1                 & 0.263             & 0.892             & 1                 & 1                 & 0.623             & 1                 & 1                 & 1                 \\ \hline
\multirow{3}{*}{$T_{max}^{\mathrm{dp}}$} & (400,200) & 0.054             & 0.057             & 0.043             & 0.061             & 0.174             & 0.601             & 0.999             & 1                 & 0.482             & 0.967             & 1                 & 1                 & 0.169             & 0.557             & 1                 & 1                 & 0.431             & 0.985             & 1                 & 1                 \\
                           & (600,300) & 0.055             & 0.057             & 0.048             & 0.055             & 0.207             & 0.779             & 1                 & 1                 & 0.576             & 0.996             & 1                 & 1                 & 0.235             & 0.738             & 1                 & 1                 & 0.600             & 1                 & 1                 & 1                 \\
                           & (800,400) & 0.050             & 0.048             & 0.059             & 0.056             & 0.240             & 0.887             & 1                 & 1                 & 0.655             & 1                 & 1                 & 1                 & 0.273             & 0.854             & 1                 & 1                 & 0.725             & 1                 & 1                 & 1                 \\ \hline
\end{tabular}
\vspace{1.5em}

\begin{tabular}{c|c|cccc|cccc|cccc|cccc|cccc}
\hline
                           &           & \multicolumn{4}{c|}{$\delta = 0$}                                             & \multicolumn{4}{c|}{$\delta = -0.25$}                                         & \multicolumn{4}{c|}{$\delta = -0.5$}                                          & \multicolumn{4}{c|}{$\delta = 0.25$}                                          & \multicolumn{4}{c}{$\delta = 0.5$}                                            \\ \hline
                           & $(n,d)$   & $\varepsilon = 1$ & $\varepsilon = 2$ & $\varepsilon = 4$ & $\varepsilon = 8$ & $\varepsilon = 1$ & $\varepsilon = 2$ & $\varepsilon = 4$ & $\varepsilon = 8$ & $\varepsilon = 1$ & $\varepsilon = 2$ & $\varepsilon = 4$ & $\varepsilon = 8$ & $\varepsilon = 1$ & $\varepsilon = 2$ & $\varepsilon = 4$ & $\varepsilon = 8$ & $\varepsilon = 1$ & $\varepsilon = 2$ & $\varepsilon = 4$ & $\varepsilon = 8$ \\ \hline
\multirow{3}{*}{$T_1^{\mathrm{dp}}$}     & (400,400) & 0.056             & 0.052             & 0.053             & 0.057             & 0.228             & 0.317             & 0.202             & 0.112             & 0.443             & 0.200             & 0.174             & 0.571             & 0.174             & 0.517             & 0.663             & 0.561             & 0.372             & 0.957             & 1                 & 1                 \\
                           & (600,600) & 0.051             & 0.055             & 0.053             & 0.053             & 0.320             & 0.462             & 0.260             & 0.120             & 0.594             & 0.256             & 0.194             & 0.747             & 0.216             & 0.679             & 0.854             & 0.754             & 0.437             & 0.994             & 1                 & 1                 \\
                           & (800,800) & 0.054             & 0.053             & 0.057             & 0.041             & 0.404             & 0.562             & 0.324             & 0.135             & 0.710             & 0.275             & 0.255             & 0.846             & 0.267             & 0.802             & 0.933             & 0.872             & 0.539             & 0.999             & 1                 & 1                 \\ \hline
\multirow{3}{*}{$T_2^{\mathrm{dp}}$}     & (400,400) & 0.046             & 0.049             & 0.058             & 0.053             & 0.209             & 0.642             & 1                 & 1                 & 0.540             & 0.985             & 1                 & 1                 & 0.094             & 0.363             & 0.996             & 1                 & 0.148             & 0.744             & 1                 & 1                 \\
                           & (600,600) & 0.051             & 0.043             & 0.054             & 0.056             & 0.221             & 0.792             & 1                 & 1                 & 0.637             & 0.996             & 1                 & 1                 & 0.098             & 0.501             & 1                 & 1                 & 0.158             & 0.876             & 1                 & 1                 \\
                           & (800,800) & 0.048             & 0.046             & 0.052             & 0.058             & 0.241             & 0.882             & 1                 & 1                 & 0.689             & 1                 & 1                 & 1                 & 0.111             & 0.607             & 1                 & 1                 & 0.193             & 0.956             & 1                 & 1                 \\ \hline
\multirow{3}{*}{$T_3^{\mathrm{dp}}$}     & (400,400) & 0.050             & 0.052             & 0.057             & 0.048             & 0.285             & 0.720             & 1                 & 1                 & 0.656             & 0.995             & 1                 & 1                 & 0.131             & 0.459             & 0.985             & 1                 & 0.251             & 0.885             & 1                 & 1                 \\
                           & (600,600) & 0.056             & 0.041             & 0.052             & 0.057             & 0.313             & 0.872             & 1                 & 1                 & 0.767             & 0.999             & 1                 & 1                 & 0.152             & 0.616             & 1                 & 1                 & 0.286             & 0.972             & 1                 & 1                 \\
                           & (800,800) & 0.054             & 0.043             & 0.056             & 0.050             & 0.355             & 0.946             & 1                 & 1                 & 0.827             & 1                 & 1                 & 1                 & 0.181             & 0.748             & 1                 & 1                 & 0.363             & 0.995             & 1                 & 1                 \\ \hline
\multirow{3}{*}{$T_{max}^{\mathrm{dp}}$} & (400,400) & 0.051             & 0.045             & 0.064             & 0.053             & 0.291             & 0.688             & 1                 & 1                 & 0.707             & 0.992             & 1                 & 1                 & 0.158             & 0.510             & 0.989             & 1                 & 0.317             & 0.939             & 1                 & 1                 \\
                           & (600,600) & 0.058             & 0.048             & 0.054             & 0.057             & 0.343             & 0.845             & 1                 & 1                 & 0.828             & 0.998             & 1                 & 1                 & 0.182             & 0.673             & 1                 & 1                 & 0.377             & 0.990             & 1                 & 1                 \\
                           & (800,800) & 0.051             & 0.049             & 0.060             & 0.045             & 0.408             & 0.932             & 1                 & 1                 & 0.898             & 1                 & 1                 & 1                 & 0.214             & 0.794             & 1                 & 1                 & 0.468             & 0.999             & 1                 & 1                 \\ \hline
\end{tabular}

\vspace{1.5em}

\begin{tabular}{c|c|cccc|cccc|cccc|cccc|cccc}
\hline
                           &            & \multicolumn{4}{c|}{$\delta = 0$}                                             & \multicolumn{4}{c|}{$\delta = -0.25$}                                         & \multicolumn{4}{c|}{$\delta = -0.5$}                                          & \multicolumn{4}{c|}{$\delta = 0.25$}                                          & \multicolumn{4}{c}{$\delta = 0.5$}                                            \\ \hline
                           & $(n,d)$    & $\varepsilon = 1$ & $\varepsilon = 2$ & $\varepsilon = 4$ & $\varepsilon = 8$ & $\varepsilon = 1$ & $\varepsilon = 2$ & $\varepsilon = 4$ & $\varepsilon = 8$ & $\varepsilon = 1$ & $\varepsilon = 2$ & $\varepsilon = 4$ & $\varepsilon = 8$ & $\varepsilon = 1$ & $\varepsilon = 2$ & $\varepsilon = 4$ & $\varepsilon = 8$ & $\varepsilon = 1$ & $\varepsilon = 2$ & $\varepsilon = 4$ & $\varepsilon = 8$ \\ \hline
\multirow{3}{*}{$T_1^{\mathrm{dp}}$}     & (400,2000) & 0.055             & 0.049             & 0.051             & 0.057             & 0.420             & 0.948             & 1                 & 1                 & 0.893             & 1                 & 1                 & 1                 & 0.162             & 0.647             & 1                 & 1                 & 0.280             & 0.940             & 1                 & 1                 \\
                           & (600,3000) & 0.046             & 0.051             & 0.050             & 0.046             & 0.510             & 0.989             & 1                 & 1                 & 0.950             & 1                 & 1                 & 1                 & 0.181             & 0.800             & 1                 & 1                 & 0.366             & 0.989             & 1                 & 1                 \\
                           & (800,4000) & 0.051             & 0.061             & 0.047             & 0.050             & 0.567             & 0.998             & 1                 & 1                 & 0.980             & 1                 & 1                 & 1                 & 0.213             & 0.885             & 1                 & 1                 & 0.421             & 0.998             & 1                 & 1                 \\ \hline
\multirow{3}{*}{$T_2^{\mathrm{dp}}$}     & (400,2000) & 0.055             & 0.054             & 0.055             & 0.058             & 0.230             & 0.757             & 1                 & 1                 & 0.694             & 0.999             & 1                 & 1                 & 0.088             & 0.327             & 0.992             & 1                 & 0.119             & 0.640             & 1                 & 1                 \\
                           & (600,3000) & 0.051             & 0.057             & 0.045             & 0.047             & 0.273             & 0.881             & 1                 & 1                 & 0.768             & 1                 & 1                 & 1                 & 0.092             & 0.456             & 1                 & 1                 & 0.155             & 0.805             & 1                 & 1                 \\
                           & (800,4000) & 0.055             & 0.058             & 0.051             & 0.055             & 0.293             & 0.939             & 1                 & 1                 & 0.835             & 1                 & 1                 & 1                 & 0.110             & 0.558             & 1                 & 1                 & 0.164             & 0.885             & 1                 & 1                 \\ \hline
\multirow{3}{*}{$T_3^{\mathrm{dp}}$}     & (400,2000) & 0.053             & 0.053             & 0.049             & 0.057             & 0.389             & 0.931             & 1                 & 1                 & 0.865             & 1                 & 1                 & 1                 & 0.145             & 0.592             & 1                 & 1                 & 0.245             & 0.918             & 1                 & 1                 \\
                           & (600,3000) & 0.050             & 0.054             & 0.048             & 0.051             & 0.477             & 0.984             & 1                 & 1                 & 0.927             & 1                 & 1                 & 1                 & 0.167             & 0.760             & 1                 & 1                 & 0.328             & 0.984             & 1                 & 1                 \\
                           & (800,4000) & 0.050             & 0.053             & 0.048             & 0.051             & 0.525             & 0.997             & 1                 & 1                 & 0.960             & 1                 & 1                 & 1                 & 0.193             & 0.852             & 1                 & 1                 & 0.373             & 0.995             & 1                 & 1                 \\ \hline
\multirow{3}{*}{$T_{max}^{\mathrm{dp}}$} & (400,2000) & 0.054             & 0.053             & 0.052             & 0.059             & 0.383             & 0.935             & 1                 & 1                 & 0.877             & 1                 & 1                 & 1                 & 0.138             & 0.584             & 1                 & 1                 & 0.235             & 0.920             & 1                 & 1                 \\
                           & (600,3000) & 0.046             & 0.054             & 0.045             & 0.048             & 0.473             & 0.983             & 1                 & 1                 & 0.940             & 1                 & 1                 & 1                 & 0.161             & 0.753             & 1                 & 1                 & 0.317             & 0.983             & 1                 & 1                 \\
                           & (800,4000) & 0.053             & 0.056             & 0.046             & 0.053             & 0.523             & 0.998             & 1                 & 1                 & 0.971             & 1                 & 1                 & 1                 & 0.187             & 0.855             & 1                 & 1                 & 0.369             & 0.997             & 1                 & 1                 \\ \hline
\end{tabular}
\end{table}
\end{landscape}

From Tables \ref{table_model_I}-\ref{table_model_II}, it can be observed that under the null hypothesis, our proposed privatized testing method effectively controls Type I errors across different settings. Notably, the empirical sizes remain robust against the different values of privacy parameter $\varepsilon$. Under the alternative hypotheses, as the signal of the test enhances, our proposed test exhibits reasonable power performance, although it depends on privacy parameters. For sufficiently large $(d,n)$, the empirical powers approach one, which confirms the consistency of the test. On the other hand, the likelihood ratio statistic $T_1^{\mathrm{dp}}$ is as powerful as $T_{max}^{\mathrm{dp}}$ in most cases. However, in scenarios where $T_1^{\mathrm{dp}}$ exhibits low power, such as $\delta = 0.25, y = 0.5$ in Table~\ref{table_model_I} and $\delta = -0.25, y = 1$ in Table~\ref{table_model_II},  $T_{max}^{\mathrm{dp}}$ effectively inherits the strong power performance of $T_2^{\mathrm{dp}}$ or $T_3^{\mathrm{dp}}$, providing a more reliable result in such cases.
   
To further examine the power performance of privatized test against the different covariance structure, we design three additional alternatives with the following structures:
	 \begin{enumerate}
	 	\item[(1)] $
	 	H_1: \mathbf{\Sigma} = \text{diag}\{1, 0.05, 0.05, 0.05, \cdots\}$, we denote
its power by Power $1$;
	 	\item[(2)] $
	 	H_1 : \mathbf{\Sigma} = (\sigma_{ij}) $ with $ \sigma_{ij} = 2^{-|i-j|}\, \text{ for } 1\leq i,j\leq d
	 	$, we denote
its power by Power $2$;

   \item[(3)] $H_1: \mathbf{\Sigma} = \text{diag}\{\underbrace{2, \cdots,2}_{d/2}, \underbrace{0.5,\cdots, 0.5}_{d/2}\}$, we denote
its power by Power $3$. 
	 \end{enumerate}
     In which, the scenario (1) was adopted in~\citep{bai2009corrections}. The empirical powers are reported in Tables~\ref{additional_table_model_I}-\ref{additional_table_model_II}.
    
\begin{table}[H]
\tiny
		 	\centering
            \setlength
		 	\tabcolsep{7pt}
		 	\renewcommand{\arraystretch}{1.35}
            \caption{Empirical powers for Model I}
            \label{additional_table_model_I}
\begin{tabular}{c|c|cccc|cccc|cccc}
\hline
                                         &            & \multicolumn{4}{c|}{Power 1}                                                  & \multicolumn{4}{c|}{Power 2}                                                  & \multicolumn{4}{c}{Power 3}                                                   \\ \hline
                                         & $(n,d)$    & $\varepsilon = 1$ & $\varepsilon = 2$ & $\varepsilon = 4$ & $\varepsilon = 8$ & $\varepsilon = 1$ & $\varepsilon = 2$ & $\varepsilon = 4$ & $\varepsilon = 8$ & $\varepsilon = 1$ & $\varepsilon = 2$ & $\varepsilon = 4$ & $\varepsilon = 8$ \\ \hline
\multirow{3}{*}{$T_{1}^{\mathrm{dp}}$}   & (400,200)  & 1            & 1                 & 1                 & 1                 & 0.5075            & 0.951             & 0.9995            & 1                 & 0.572             & 0.945             & 1                 & 1                 \\
                                         & (400,200)  & 1                 & 1                 & 1                 & 1                 & 0.657             & 0.995             & 1                 & 1                 & 0.748             & 0.994             & 1                 & 1                 \\
                                         & (400,200)  & 1                 & 1                 & 1                 & 1                 & 0.7635            & 0.999             & 1                 & 1                 & 0.845             & 1                 & 1                 & 1                 \\ \hline
\multirow{3}{*}{$T_{2}^{\mathrm{dp}}$}   & (400,200)  & 0.997             & 1                 & 1                 & 1                 & 0.496             & 0.999             & 1                 & 1                 & 0.424             & 0.999             & 1                 & 1                 \\
                                         & (400,200)  & 1                 & 1                 & 1                 & 1                 & 0.627             & 1                 & 1                 & 1                 & 0.551             & 1                 & 1                 & 1                 \\
                                         & (400,200)  & 1                 & 1                 & 1                 & 1                 & 0.730             & 1                 & 1                 & 1                 & 0.656             & 1                 & 1                 & 1                 \\ \hline
\multirow{3}{*}{$T_{3}^{\mathrm{dp}}$}   & (400,200)  & 1                 & 1                 & 1                 & 1                 & 0.596             & 0.997             & 1                 & 1                 & 0.599             & 0.999             & 1                 & 1                 \\
                                         & (400,200)  & 1                 & 1                 & 1                 & 1                 & 0.749             & 1                 & 1                 & 1                 & 0.771             & 1                 & 1                 & 1                 \\
                                         & (400,200)  & 1                 & 1                 & 1                 & 1                 & 0.858             & 1                 & 1                 & 1                 & 0.869             & 1                 & 1                 & 1                 \\ \hline
\multirow{3}{*}{$T_{max}^{\mathrm{dp}}$} & (400,200)  & 1                 & 1                 & 1                 & 1                 & 0.598             & 0.999             & 1                 & 1                 & 0.609             & 1                 & 1                 & 1                 \\
                                         & (600,300)  & 1                 & 1                 & 1                 & 1                 & 0.756             & 1                 & 1                 & 1                 & 0.778             & 1                 & 1                 & 1                 \\
                                         & (800,400)  & 1                 & 1                 & 1                 & 1                 & 0.856             & 1                 & 1                 & 1                 & 0.869             & 1                 & 1                 & 1                 \\ \hline\hline
\multirow{3}{*}{$T_{1}^{\mathrm{dp}}$}   & (400,400)  & 0.815             & 1                 & 1                 & 1                 & 0.199             & 0.595             & 0.932             & 0.999             & 0.334             & 0.914             & 0.999             & 1                 \\
                                         & (600,600)  & 0.905             & 1                 & 1                 & 1                 & 0.220             & 0.747             & 0.988             & 1                 & 0.412             & 0.980             & 1                 & 1                 \\
                                         & (800,800)  & 0.928             & 1                 & 1                 & 1                 & 0.262             & 0.843             & 0.998             & 1                 & 0.504             & 0.994             & 1                 & 1                 \\ \hline
\multirow{3}{*}{$T_{2}^{\mathrm{dp}}$}   & (400,400)  & 0.457             & 0.933             & 1                 & 1                 & 0.179             & 0.694             & 1                 & 1                 & 0.179             & 0.817             & 1                 & 1                 \\
                                         & (600,600)  & 0.615             & 0.987             & 1                 & 1                 & 0.213             & 0.836             & 1                 & 1                 & 0.225             & 0.934             & 1                 & 1                 \\
                                         & (800,800)  & 0.681             & 0.995             & 1                 & 1                 & 0.223             & 0.896             & 1                 & 1                 & 0.254             & 0.968             & 1                 & 1                 \\ \hline
\multirow{3}{*}{$T_{3}^{\mathrm{dp}}$}   & (400,400)  & 0.599             & 0.996             & 1                 & 1                 & 0.210             & 0.636             & 0.996             & 1                 & 0.280             & 0.886             & 1                 & 1                 \\
                                         & (600,600)  & 0.741             & 1                 & 1                 & 1                 & 0.249             & 0.799             & 1                 & 1                 & 0.348             & 0.971             & 1                 & 1                 \\
                                         & (800,800)  & 0.815             & 1                 & 1                 & 1                 & 0.289             & 0.873             & 1                 & 1                 & 0.418             & 0.989             & 1                 & 1                 \\ \hline
\multirow{3}{*}{$T_{max}^{\mathrm{dp}}$} & (400,400)  & 0.803             & 1                 & 1                 & 1                 & 0.211             & 0.707             & 1                 & 1                 & 0.300             & 0.912             & 1                 & 1                 \\
                                         & (600,600)  & 0.898             & 1                 & 1                 & 1                 & 0.246             & 0.852             & 1                 & 1                 & 0.371             & 0.982             & 1                 & 1                 \\
                                         & (800,800)  & 0.929             & 1                 & 1                 & 1                 & 0.281             & 0.916             & 1                 & 1                 & 0.457             & 0.996             & 1                 & 1                 \\ \hline\hline
\multirow{3}{*}{$T_{1}^{\mathrm{dp}}$}   & (400,2000) & 0.999             & 1                 & 1                 & 1                 & 0.070             & 0.076             & 0.155             & 0.333             & 0.192             & 0.704             & 1                 & 1                 \\
                                         & (600,3000) & 1                 & 1                 & 1                 & 1                 & 0.051             & 0.092             & 0.197             &  0.446                 & 0.236             & 0.854             & 1                 & 1                 \\
                                         & (800,4000) & 1                 & 1                 & 1                 & 1                 & 0.056             & 0.118             & 0.250             & 0.536                  & 0.275             & 0.933             & 1                 & 1                 \\ \hline
\multirow{3}{*}{$T_{2}^{\mathrm{dp}}$}   & (400,2000) & 0.970             & 1                 & 1                 & 1                 & 0.063             & 0.078             & 0.360             & 0.953             & 0.116             & 0.421             & 1                 & 1                 \\
                                         & (600,3000) & 0.993             & 1                 & 1                 & 1                 & 0.059             & 0.109             & 0.475             &  0.998                 & 0.106             & 0.572             & 1                 & 1                 \\
                                         & (800,4000) & 0.999             & 1                 & 1                 & 1                 & 0.059             & 0.126             & 0.584             &           1        & 0.133             & 0.679             & 1                 & 1                 \\ \hline
\multirow{3}{*}{$T_{3}^{\mathrm{dp}}$}   & (400,2000) & 0.993             & 1                 & 1                 & 1                 & 0.064             & 0.071             & 0.136             & 0.132             & 0.176             & 0.675             & 1                 & 1                 \\
                                         & (600,3000) & 0.999             & 1                 & 1                 & 1                 & 0.050             & 0.094             & 0.155             &  0.159                 & 0.218             & 0.833             & 1                 & 1                 \\
                                         & (800,4000) & 1                 & 1                 & 1                 & 1                 & 0.054             & 0.115             & 0.199             &              0.186     & 0.249             & 0.911             & 1                 & 1                 \\ \hline
\multirow{3}{*}{$T_{max}^{\mathrm{dp}}$} & (400,2000) & 0.997             & 1                 & 1                 & 1                 & 0.063             & 0.083             & 0.302             & 0.939             & 0.169             & 0.665             & 1                 & 1                 \\
                                         & (600,3000) & 1                 & 1                 & 1                 & 1                 & 0.060             & 0.104             & 0.408             & 0.994                  & 0.199             & 0.824             & 1                 & 1                 \\
                                         & (800,4000) & 1                 & 1                 & 1                 & 1                 & 0.060             & 0.136             & 0.529             &  1                 & 0.242             & 0.912             & 1                 & 1                 \\ \hline
\end{tabular}
\end{table}

\begin{table}[H]
\tiny
		 	\centering
            \setlength
		 	\tabcolsep{7pt}
		 	\renewcommand{\arraystretch}{1.35}
            \caption{Empirical powers for Model II}
            \label{additional_table_model_II}
            \begin{tabular}{c|c|cccc|cccc|cccc}
\hline
                                         &            & \multicolumn{4}{c|}{Power 1}                                                  & \multicolumn{4}{c|}{Power 2}                                                  & \multicolumn{4}{c}{Power 3}                                                   \\ \hline
                                         & $(n,d)$    & $\varepsilon = 1$ & $\varepsilon = 2$ & $\varepsilon = 4$ & $\varepsilon = 8$ & $\varepsilon = 1$ & $\varepsilon = 2$ & $\varepsilon = 4$ & $\varepsilon = 8$ & $\varepsilon = 1$ & $\varepsilon = 2$ & $\varepsilon = 4$ & $\varepsilon = 8$ \\ \hline
\multirow{3}{*}{$T_{1}^{\mathrm{dp}}$}   & (400,200)  & 1                 & 1                 & 1                 & 1                 & 0.492             & 0.954             & 1                 & 1                 & 0.569             & 0.949             & 1                 & 1                 \\
                                         & (400,200)  & 1                 & 1                 & 1                 & 1                 & 0.653             & 0.999             & 1                 & 1                 & 0.747             & 0.994             & 1                 & 1                 \\
                                         & (400,200)  & 1                 & 1                 & 1                 & 1                 & 0.769             & 1                 & 1                 & 1                 & 0.838             & 1                 & 1                 & 1                 \\ \hline
\multirow{3}{*}{$T_{2}^{\mathrm{dp}}$}   & (400,200)  & 0.998             & 1                 & 1                 & 1                 & 0.492             & 0.998             & 1                 & 1                 & 0.421             & 0.998             & 1                 & 1                 \\
                                         & (400,200)  & 1                 & 1                 & 1                 & 1                 & 0.625             & 1                 & 1                 & 1                 & 0.546             & 1                 & 1                 & 1                 \\
                                         & (400,200)  & 1                 & 1                 & 1                 & 1                 & 0.731             & 1                 & 1                 & 1                 & 0.654             & 1                 & 1                 & 1                 \\ \hline
\multirow{3}{*}{$T_{3}^{\mathrm{dp}}$}   & (400,200)  & 1                 & 1                 & 1                 & 1                 & 0.605             & 0.998             & 1                 & 1                 & 0.599             & 1                 & 1                 & 1                 \\
                                         & (400,200)  & 1                 & 1                 & 1                 & 1                 & 0.736             & 1                 & 1                 & 1                 & 0.771             & 1                 & 1                 & 1                 \\
                                         & (400,200)  & 1                 & 1                 & 1                 & 1                 & 0.859             & 1                 & 1                 & 1                 & 0.868             & 1                 & 1                 & 1                 \\ \hline
\multirow{3}{*}{$T_{max}^{\mathrm{dp}}$} & (400,200)  & 1                 & 1                 & 1                 & 1                 & 0.603             & 1                 & 1                 & 1                 & 0.604             & 1                 & 1                 & 1                 \\
                                         & (600,300)  & 1                 & 1                 & 1                 & 1                 & 0.753             & 1                 & 1                 & 1                 & 0.774             & 1                 & 1                 & 1                 \\
                                         & (800,400)  & 1                 & 1                 & 1                 & 1                 & 0.863             & 1                 & 1                 & 1                 & 0.869             & 1                 & 1                 & 1                 \\ \hline \hline 
\multirow{3}{*}{$T_{1}^{\mathrm{dp}}$}   & (400,400)  & 0.816             & 1                 & 1                 & 1                 & 0.194             & 0.598             & 0.926             & 0.997             & 0.333             & 0.891             & 0.999             & 1                 \\
                                         & (600,600)  & 0.892             & 1                 & 1                 & 1                 & 0.215             & 0.750             & 0.990             & 1                 & 0.417             & 0.980             & 1                 & 1                 \\
                                         & (800,800)  & 0.950             & 1                 & 1                 & 1                 & 0.275             & 0.858             & 1                 & 1                 & 0.501             & 0.995             & 1                 & 1                 \\ \hline
\multirow{3}{*}{$T_{2}^{\mathrm{dp}}$}   & (400,400)  & 0.455             & 0.934             & 1                 & 1                 & 0.179             & 0.677             & 1                 & 1                 & 0.186             & 0.790             & 1                 & 1                 \\
                                         & (600,600)  & 0.612             & 0.983             & 1                 & 1                 & 0.186             & 0.831             & 1                 & 1                 & 0.222             & 0.928             & 1                 & 1                 \\
                                         & (800,800)  & 0.708             & 0.997             & 1                 & 1                 & 0.229             & 0.910             & 1                 & 1                 & 0.247             & 0.973             & 1                 & 1                 \\ \hline
\multirow{3}{*}{$T_{3}^{\mathrm{dp}}$}   & (400,400)  & 0.595             & 0.991             & 1                 & 1                 & 0.206             & 0.620             & 0.996             & 1                 & 0.268             & 0.878             & 1                 & 1                 \\
                                         & (600,600)  & 0.734             & 1                 & 1                 & 1                 & 0.231             & 0.787             & 1                 & 1                 & 0.348             & 0.964             & 1                 & 1                 \\
                                         & (800,800)  & 0.837             & 1                 & 1                 & 1                 & 0.288             & 0.887             & 1                 & 1                 & 0.406             & 0.993             & 1                 & 1                 \\ \hline
\multirow{3}{*}{$T_{max}^{\mathrm{dp}}$} & (400,400)  & 0.796             & 1                 & 1                 & 1                 & 0.209             & 0.694             & 1                 & 1                 & 0.295             & 0.896             & 1                 & 1                 \\
                                         & (600,600)  & 0.882             & 1                 & 1                 & 1                 & 0.226             & 0.856             & 1                 & 1                 & 0.371             & 0.979             & 1                 & 1                 \\
                                         & (800,800)  & 0.946             & 1                 & 1                 & 1                 & 0.287             & 0.927             & 1                 & 1                 & 0.446             & 0.998             & 1                 & 1                 \\ \hline\hline 
\multirow{3}{*}{$T_{1}^{\mathrm{dp}}$}   & (400,2000) & 0.998             & 1                 & 1                 & 1                 & 0.065             & 0.075             & 0.152             & 0.360             & 0.183             & 0.719             & 1                 & 1                 \\
                                         & (600,3000) & 1                 & 1                 & 1                 & 1                 & 0.054             & 0.080             & 0.192             & 0.452             & 0.221             & 0.868             & 1                 & 1                 \\
                                         & (800,4000) & 1                 & 1                 & 1                 & 1                 & 0.052             & 0.103             & 0.252             & 0.559             & 0.253             & 0.936             & 1                 & 1                 \\ \hline
\multirow{3}{*}{$T_{2}^{\mathrm{dp}}$}   & (400,2000) & 0.971             & 1                 & 1                 & 1                 & 0.074             & 0.091             & 0.354             & 0.963             & 0.101             & 0.442             & 0.999             & 1                 \\
                                         & (600,3000) & 0.994             & 1                 & 1                 & 1                 & 0.060             & 0.096             & 0.455             & 0.993             & 0.106             & 0.594             & 1                 & 1                 \\
                                         & (800,4000) & 1.000             & 1                 & 1                 & 1                 & 0.052             & 0.116             & 0.577             & 0.999             & 0.134             & 0.688             & 1                 & 1                 \\ \hline
\multirow{3}{*}{$T_{3}^{\mathrm{dp}}$}   & (400,2000) & 0.993             & 1                 & 1                 & 1                 & 0.065             & 0.079             & 0.129             & 0.144             & 0.165             & 0.685             & 1                 & 1                 \\
                                         & (600,3000) & 0.999             & 1                 & 1                 & 1                 & 0.052             & 0.084             & 0.154             & 0.176             & 0.202             & 0.838             & 1                 & 1                 \\
                                         & (800,4000) & 1                 & 1                 & 1                 & 1                 & 0.051             & 0.106             & 0.196             & 0.195             & 0.232             & 0.918             & 1                 & 1                 \\ \hline
\multirow{3}{*}{$T_{max}^{\mathrm{dp}}$} & (400,2000) & 0.997             & 1                 & 1                 & 1                 & 0.069             & 0.090             & 0.303             & 0.949             & 0.152             & 0.678             & 1                 & 1                 \\
                                         & (600,3000) & 1                 & 1                 & 1                 & 1                 & 0.056             & 0.094             & 0.383             & 0.991             & 0.187             & 0.838             & 1                 & 1                 \\
                                         & (800,4000) & 1                 & 1                 & 1                 & 1                 & 0.055             & 0.114             & 0.513             & 1                 & 0.222             & 0.918             & 1                 & 1                 \\ \hline
\end{tabular}
\end{table}
From the results reported in Tables~\ref{additional_table_model_I}-\ref{additional_table_model_II}, we observe that under additional alternatives, our proposed test shows a similar pattern of increasing empirical powers as $(d,n)$ grows larger.  It is noteworthy that even for small $(d,n)$ and high privacy level (e.g. $\varepsilon=1$ for Power 1), the privatized test still maintains satisfactory power levels. This suggests that the privatized test is robust against the different covariance structures. For $y = 0.5$, we also compare the privatized test with the methods proposed in~\citep{bai2009corrections, jiang2012likelihood}, and we find their empirical power reaches one in all settings. Although for high privacy levels ($\varepsilon = 1, 2$), the empirical power of the privatized test catches up and aligns with the performance of CLR test. Similarly to the observations above, in scenarios where $T_1^{\mathrm{dp}}$ exhibits low power, such as Power 2 with $y = 0.5$ in Table~\ref{additional_table_model_I} and Power 2 with $y = 5$ in Table~\ref{additional_table_model_II}, $T_{max}^{\mathrm{dp}}$ eliminates low power due to strong power performance of $T_2^{\mathrm{dp}}$ or $T_3^{\mathrm{dp}}$, which illustrates the integrated strategy of necessity once again.

\subsection{Real data analysis}

 In this subsection, we apply our proposed test to the sonar
 data, which is available in the UCI database (\url{https://archive.ics.uci.edu/}).
 The sonar dataset is a collection of sonar returns, which are the amplitudes of bouncing signals off underwater objects. This dataset was analyzed  by~\citep{yi2013sure} and \citep{qiu2015bandwidth}. It consists of $208$ observations, $60$ input variables, and one output variable, where each observation represents a sonar return. These returns were obtained by bouncing signals off a metal cylinder and a cylindrically shaped rock, with $111$ sonar returns from the metal cylinder and $97$ returns from the rock.
 
 In the study conducted by~\citep{yi2013sure}, it was found that the sample covariance matrix of the $60$ input variables exhibited a decaying pattern: the diagonal values decreased noticeably as the distance from the diagonal increased. Building upon this finding, we are interested in testing whether the sample correlation matrix $\mathbf{R}$ of this dataset is equal to the identity matrix $\mathbf{I}_d$. Specifically, we conduct the following hypothesis testing as:
 \begin{equation}
 	\label{correlation test 1}
 	H_0: \mathbf{R} = \mathbf{I}_d \quad \mathrm{versus} \quad H_1: \mathbf{R} \neq \mathbf{I}_d.
 \end{equation}
Here the privacy parameter $\varepsilon$ takes values from $\{1, 2, 4, 8\}$. To verify the results, we also use the CLR statistic proposed in~\citep{bai2009corrections}, which is consistent with our findings.  From Table \ref{Table: 4.5}, we observe that all the $p$-values are extremely close to zero, leading us to reject the null hypothesis $H_0: \mathbf{R} = \mathbf{I}_d$. Therefore, it is reasonable to conclude  that the covairates are correlated with each other.  This result aligns with~\citep{zhubandwidth}, which suggests that the covariance matrix exhibits a potentially bundleable structure.

 \begin{table}[H]
 	\renewcommand{\arraystretch}{1.3}
 	\centering
 	\caption{$p$-values for the sonar data}
 	\begin{tabular}{ccccc}
\hline
                        & $\varepsilon  = 1$ & $\varepsilon  = 2$ & $\varepsilon  = 4$ & $\varepsilon  = 8$ \\ \hline
$T_1^{\mathrm{dp}}$     & 3.24e-09           & 0                  & 0                  & 0                  \\
$T_2^{\mathrm{dp}}$     & 0                  & 0                  & 0                  & 0                  \\
$T_3^{\mathrm{dp}}$     & 5.80e-13           & 0                  & 0                  & 0                  \\
$T_{max}^{\mathrm{dp}}$ & 0                  & 0                  & 0                  & 0                  \\ \hline
\end{tabular}
 \label{Table: 4.5}
 \end{table}

\section{Concluding remarks}
\label{section:5}
In this paper, we propose a privatized test statistic using classical privacy mechanisms for high-dimensional covariance matrices within the differential privacy framework.  To strike a balanced trade-off between statistical accuracy and privacy protection, we further characterize the global sensitivity of sample eigenvalues for sub-Gaussian populations to determine a suitable noise scale. The constructed test can effectively control the Type I error while exhibiting satisfactory power performance. Our analysis provides new insights for other types of statistical tasks with privacy constraints, particularly in the context of unbounded populations. We also points to some limitations and intriguing directions for future investigation.

First, throughout our discussion, the ratio of dimension to sample size is of constant order, namely, $y \in (0, \infty)$. As claimed in Theorem \ref{upper bound of eigenvalue sensitivity}, to achieve privacy protection with high probability, the noise assigned for sample eigenvalues must have a standard deviation of $d/n$-order. Therefore, the proposed statistics and privacy analysis are not valid for the case of $d/n \to \infty$. It would be an interesting question that how to construct an appropriate statistic in this regime. Secondly, our work concentrates on the  one-sample test for  covariance matrice. There exists extensive research on comparing  covariance matrices between two or multiple populations, which is also attractive in contemporary applications \citep{bai2009corrections, jiang2012likelihood, zheng2015substitution}. It is therefore of interest to extend our method to two-sample or multi-sample covariance tests under DP. Thirdly, for other complex data types that frequently appear in real-world applications, such as time series~\citep{hamilton2020time}, panel data~\citep{hsiao2022analysis}, privacy concerns should be given greater consideration when performing statistical inference on these data types. Constructing privacy-preserving statistics tailored to specific data structures is another interesting and challenging direction. 

\bibliographystyle{apalike}
\bibliography{ref}

\appendix

\section{Proof of main results}
\label{app: A}
\subsection{Proof of Theorem \ref{upper bound of eigenvalue sensitivity}}

We frist introduce a lemma about tail bound for quadratic forms of sub-Gaussian vectors, which can be found in~\citep{hsu2012tail}.

\begin{lemma}\label{subgaussian tail bound}
Let $\mathbf A \in \mathbb{R}^{n \times n} $ be a matrix, and ${\mathbf \Sigma} = \mathbf A^{\top}\mathbf A$. Suppose that $\mathbf {z} = (z_1, \cdots ,z_n)^{\top}$ be a random vector such that, for some $\boldsymbol{\mu} \in \mathbb R^n$ and $\sigma > 0$, 
\begin{equation*}
\mathbb E \left[ \exp\left\{ \boldsymbol{\alpha}^{\top}(\mathbf z - \boldsymbol{\mu}) \right\}\right] \leq \exp\left( \frac{\|\boldsymbol{\alpha}\|_2^2 \sigma^2}{2}\right)
\end{equation*}
for all $\boldsymbol{\alpha} \in \mathbb R^n$. Then for all $t>0$, 
\begin{equation*}
\mathbb{P}\left(\|\mathbf{Az}\|_2^2>\sigma^2\cdot\left\{\mathrm{Tr} (\mathbf{\Sigma})+2\sqrt{\mathrm{Tr}(\mathbf{\Sigma}^2)t}+2\| \mathbf{\Sigma}\|t\right\}+\|\mathbf{A}\boldsymbol{\mu}\|_2^2\cdot\left[1+2\left\{\frac{\|\mathbf{\Sigma}\|^2}{\mathrm{Tr}( \mathbf{\Sigma}^2)}t\right\}^{1/2}\right]\right)\leq e^{-t}.
\end{equation*}    
\end{lemma}

\begin{proof}[Proof of Theorem~\ref{upper bound of eigenvalue sensitivity}]
Let $\mathbf X,\tilde{\mathbf X} \in \mathbb{R}^{n \times d}$ be neighboring design matrices drawn from $\mathbf x$. Following the  similar approach in~\citep{amin2019differentially}, we assume that $\mathbf X$ and $\tilde{\mathbf X}$ differ only in the first row and rewrite
\begin{equation*}
\mathbf{X} =\begin{bmatrix}
\mathbf{x}_1^{\top}\\
\mathbf{Y}
\end{bmatrix}, \quad 
\tilde{\mathbf{X}} =\begin{bmatrix}
\tilde{\mathbf{x}}_1^{\top}\\
\mathbf{Y}
\end{bmatrix},
\end{equation*}
then the sample covariance matrix defined by $\mathbf{X}$ is $ n^{-1} \mathbf{X}^{\top} \mathbf{X}=n^{-1}  \mathbf{Y}^{\top} \mathbf{Y}+n^{-1} \mathbf{x}_1 \mathbf{x}_1^{\top}$.
For any $\mathbf v \in \mathbb R^d$,
\$
\mathbf{v}^{\top}\mathbf{X}^{\top} \mathbf{X}\mathbf{v} - 	\mathbf{v}^{\top}\mathbf{Y}^{\top} \mathbf{Y}\mathbf{v} = 
\mathbf{v}^{\top}\mathbf{x}_1 \mathbf{x}_1^{\top}\mathbf{v} \geq 0,
\$
then it holds that
\$
\lambda_j \left(\frac{1}{n} \mathbf{X}^{\top} \mathbf{X}\right)  \geq \lambda_j \left(\frac{1}{n} \mathbf{Y}^{\top} \mathbf{Y}\right),\quad 1\leq j\leq d,
\$
which implies that
\$
\left \|\Lambda \left(\frac{1}{n} \mathbf{X}^{\top} \mathbf{X}\right)-\Lambda \left(\frac{1}{n} \mathbf{Y}^{\top} \mathbf{Y}\right)\right\|_1 =  \text{Tr}\left(\frac{1}{n}\mathbf{X}^{\top} \mathbf{X}-\frac{1}{n}\mathbf{Y}^{\top} \mathbf{Y}\right) = \frac{1}{n}\mathbf{x}_1^{\top} \mathbf{x}_1.
\$
The same results holds for $\tilde{\mathbf X}$ so we have
$$
\left\| \Lambda \left(\frac{1}{n}\mathbf{X}^{\top} \mathbf{X}\right)  -\Lambda \left(\frac{1}{n} \tilde{\mathbf{X}}^{\top} \tilde{\mathbf{X}}\right)\right\|_1 \leq \frac{1}{n}{\textbf{x}}_1^{\top} {\textbf{x}}_1+\frac{1}{n}\tilde{\textbf{x}}_1^{\top} \tilde{\textbf{x}}_1.
$$
Since $\mathbf x = \mathbf{\Sigma}^{1/2}\mathbf{z}$, applying Lemma \ref{subgaussian tail bound} to $\mathbf{z}$ by setting $\mathbf{A} = \mathbf{\Sigma}^{1/2}$ and $\boldsymbol{\mu} = \boldsymbol{0}$, we obtain that
	$$
	\mathbb{P}\left[ \|\mathbf{x}\|_2^2 > \sigma^2 \left\{ \mathrm{Tr} (\mathbf{\Sigma} )+ 2\sqrt{\mathrm{Tr} (\mathbf{\Sigma}^2)t} + 2 \|\mathbf{\Sigma}\|t\right\} \right] \leq e^{-t},
	$$
	where $\sigma$ denotes the sub-Gaussian parameter of $\mathbf{z}$. Replacing $t$ with $nt$, we have
	$$
	\mathbb{P}\left[ \frac{1}{n} \mathbf{x}^{\top}\mathbf{x} > \sigma^2 \left\{ \frac{ \mathrm{Tr} (\mathbf{\Sigma})}{n} + 2\sqrt{\mathrm{Tr} (\mathbf{\Sigma}^2)}\cdot \sqrt{\frac{t}{n}} + 2 \|\mathbf{\Sigma}\|t\right\}\right] \leq e^{-nt}.
	$$
	By the independence between $\mathbf{x}_1$ and $\tilde{\mathbf{x}}_1$, we conclude that 
	$$
	\mathbb P\left[\Delta_1 \Lambda \leq  \sigma^2 \cdot\left\{\frac{2\mathrm{Tr}(\mathbf{\Sigma})}{n} + 4 \sqrt{ \mathrm{Tr}(\mathbf{\Sigma}^2)} \cdot \sqrt{\frac{t}{n}}
	+4 \left\|\mathbf{\Sigma}\right\|t   \right\}  \right] \geq (1-e^{-nt})^2 \geq 1-2 e^{-nt}.
	$$
\end{proof}

\subsection{Proof of Theorem~\ref{privacy guarantee}}

\begin{proof}[Proof of~Theorem~\ref{privacy guarantee}]
	 Under Assumption  \ref{sub gaussian assumption}, we have $\mathrm{Tr}(\mathbf{\Sigma})\leq  \gamma d$ and $\mathrm{Tr}(\mathbf{\Sigma}^2) \leq c^2d$. For $0<t\leq 1$, 
	\begin{enumerate}
		\item[(1)]
		$$
		\frac{2\mathrm{Tr(\mathbf{\Sigma})}  }{n} \leq  \frac{2\gamma d}{n},
		$$
		\item[(2)]
		$$
		4\sqrt{\mathrm{Tr} (\mathbf{\Sigma}^2)}\cdot \sqrt{\frac{t}{n}} \leq 4 \sqrt{ c^2 d}\cdot \sqrt{\frac{t}{n}}
		\leq 4c\cdot \max\left\{ 1, \frac{d}{n} \right\}\cdot  \sqrt{t},$$
		\item[(3)]
		$$
		4 \|\mathbf{\Sigma}\|t \leq 4c\cdot t \leq 
		4c \cdot \max\left\{ 1, \frac{d}{n} \right\} \cdot \sqrt{t}.
		$$
	\end{enumerate}
	Taking $t = n^{-2r}$ for some $0<r<1/2$, we have
	$$
	\sigma^2 \left\{ \frac{2 \mathrm{Tr} (\mathbf{\Sigma})}{n} + 4\sqrt{\mathrm{Tr} (\mathbf{\Sigma}^2)}\cdot \sqrt{\frac{t}{n}} + 4 \|\mathbf{\Sigma}\|t\right\} \leq  \sigma^2\left(\frac{2\gamma d}{n} + 8c \cdot \max\left\{ 1, \frac{d}{n} \right\}\cdot n^{-r} \right)\leq   \frac{2.01\sigma^2\gamma d}{n},
	$$
    where the last equality follows if $n^{-r}  \leq  \frac{\gamma}{800c}\min \left\{ 1 ,\frac{d}{n}\right\}$. Combining with Theorem \ref{upper bound of eigenvalue sensitivity}, we have
    $$
    \mathbb P \left( \Delta_1 \Lambda \leq \frac{2.01\sigma^2\gamma d}{n}\right) \geq 1-2 \exp\left\{-n^{1-2r}\right\}
    $$
   provided $n^{-r}  \leq  c^{\prime} \min \left\{ 1 ,\frac{d}{n}\right\}$ with $c^{\prime} = \frac{\gamma}{800c}$.

By the definition of Laplacian mechanism in Lemma \ref{lemma Laplacian mechanism}, the privatized sample eigenvalues $(\tilde{\lambda}_1,\ldots ,\tilde{\lambda}_K)$ is $\varepsilon$-differentially private. Applying the post-processing property stated in Lemma \ref{Post-processing} yields the desired result.
\end{proof}

\subsection{Proof of Lemma~\ref{asymptotic normality}}

According to the {\it Cramer-Wold} theorem, it suffices to prove the asymptotic normality of the linear combination $ L^{\mathrm{dp}}(\mathbf v^\top\mathbf g)$ for arbitrary direction $\mathbf v=(v_1,\cdots,v_r)^\top \in \RR^r$ with $\|\mathbf v\|_1=1$. We proceed the proof in three parts ({\bf Part A-C}) as follows. \\

\noindent {\bf Part A.} We first derive the marginal asymptotic distribution for each coordinate $g_m$, $1\le m\le r$. It can be divided  into six main steps as follows.

{\bf Step 1}. Approximate the eigenvalue $\lambda_i$ by $\alpha_i$, $1\le i\le K$, where $\alpha_i$ denotes the {\it classical location} of $\lambda_i $, i.e., the $(1-i/d)$-quantile of the generalized M-P law defined via 
\$
F_{y,\rho_{\kappa}}(\alpha_i)=1-\frac{i}{d},\ {\rm{for}}\ 1\le i\le K.
\$

According to the rigidity of the eigenvalues $\{\lambda_i\}_{i = 1}^K$ in Lemma~\ref{rigidity}, we have $|\lambda_i-\alpha_i|\prec n^{-1}$ uniformly for those $i$'s satisfying $\alpha_i\in \cup_{1\le k\le p}(a_{2k},a_{2k-1})$, and $|\lambda_i-\alpha_i|\prec n^{-2/3}$ uniformly for those $i$'s satisfying sticking to the some edge $a_{2k}$ (or $a_{2k-1}$) in the sense that $|i- \sum_{1\le k\le j}d_k| \le N$ with $d_k$ defined in \eqref{def_dk}, where $k\in \{0,1,\cdots,p\}$ and $N\in \NN^+$ is any fixed integer. 
Then, for any large constant $D>0$ and small constant $\epsilon>0$, it holds that 
\$
\sup_{1\le i\le K}\PP\{|\lambda_i-\alpha_i|\ge n^{-2/3+\epsilon}\}\le n^{-D},\quad 
\$ 
for any $n\ge n_0$ with some large enough $n_0\in\NN^{+}$, where we relax the convergence rate $n^{-1}$ of the bulk eigenvalues to $n^{-2/3}$ for simplicity of presentation and this rate is enough to ensure the $n^{-1/2}$-CLT for $L^{\rm dp}(g_m)$ afterwards. Taking $D>1$, it follows that 
\#\label{step1a}
\PP \left\{ |\lambda_i-\alpha_i|\le n^{-2/3+\epsilon},\ \forall 1\le i\le K\right\} \ge 1-n^{-D}\cdot K =1-o(1)
\#
for any small $\epsilon>0$ and large enough $n$. 

Under the regularity Assumption~\ref{assumption approximation error} on $\mathbf g$, \eqref{step1a} leads to
\$
L^{\mathrm{dp}}(g_m)=\frac{1}{K} \sum_{1\le i\le K} g_m(\lambda_i+\ell_i)
= \frac{1}{K} \sum_{1\le i\le K} g_m(\alpha_i+\ell_i) + o_p(K^{-1/2}),\ 1\le m\le r,
\$
where \eqref{step1a} can match condition \eqref{approximation_condition} by taking $\epsilon<1/12$.  

{\bf Step 2}. Build the CLT for $L^{\rm dp}(g_m)$. For any fixed $1\le i\le K$, we denote
\$
\sigma_i^2(g_m):= & {\rm Var}\left\{g_m(\alpha_i+\ell_i) \right\}
= \mathbb E \left\{g_m^2(\alpha_i+\ell_i) \right\} -\left[ \mathbb E \left\{g_m(\alpha_i+\ell_i)\right\}\right]^2,
\$
where 
\$
\mathbb E\left\{g_m(\alpha_i+\ell_i)\right\} = \int g_m(\alpha_i+t)\mathrm{d} H(t)\ {\rm{and}}\ 
\mathbb E\left\{ g_m^2(\alpha_i+\ell_i)\right\} = \int g_m^2(\alpha_i+t)\mathrm{d} H(t).
\$
Using {\it Lyapunov's Central Limit Theorem} in Lemma~\ref{Lyapunov_CLT}, we have
\#\label{Lyapunov_CLT_eq}
\frac{\sqrt{K}}{s_d(g_m)}\left[\frac{1}{K} \sum_{1\le i\le K} g_m(\alpha_i+\ell_i) -\frac{1}{K}\sum_{1\le i\le K}\EE \left\{ g_m(\alpha_i+\ell_i) \right\} \right]
\overset{D}{\longrightarrow} \cN (0,1)
\#
with $s_d(g_m):=\sqrt{\sum_{1\le i\le K}\sigma_i^2(g_m)/K}$. That is, 
\$
\sqrt{K}\left[L^{\mathrm{dp}}(g_m) - \mathbb E \{L^{\mathrm{dp}}(g_m)\}\right]/s_d(g_m)\xrightarrow{D}\mathcal N (0,1).
\$ 
Therefore, it remains to calculate the mean $K^{-1}\sum_{1\le i\le K}\mathbb E \left\{g_m(\alpha_i+\ell_i)\right\}$ and variance $s_d^2(g_m)$.

{\bf Step 3}. Calculate $\sigma_i^2(g_m)$ for any fixed $1\le i\le K$. It boils down to $\mathbb E \left\{g_m^2(\alpha_i+\ell_i)\right\}$ and $\left[\mathbb E \left\{g_m(\alpha_i+\ell_i)\right\}\right]^2$. Denote the probability density function of the Laplacian noise as $h(t)$, then it holds that
\$
\mathbb E \left\{g_m^2(\alpha_i+\ell_i)\right\}
& = \int g_m^2(\alpha_i+t)\mathrm{d} H(t)
= \int g_m^2(\alpha_i+t) h(t) \mathrm{d} t
=: b_{g_m,g_m}(\alpha_i),\\
\left[\mathbb E \left\{g_m(\alpha_i+\ell_i)\right\}\right]^2
& = \left\{\int g_m(\alpha_i+t)\mathrm{d} H(t)\right\}^2
= \left\{\int g_m(\alpha_i+t) h(t) \mathrm{d} t\right\}^2
=: b_{g_m}^2(\alpha_i).
\$
It follows that 
\$
\sigma_i^2 (g_m) = \mathbb E\left\{ g_m^2(\alpha_i+\ell_i)\right\} -\left[ \mathbb E \left\{g_m(\alpha_i+\ell_i)\right\}\right]^2
= b_{g_m,g_m}(\alpha_i)-b_{g_m}^2(\alpha_i).
\$

{\bf Step 4}. Sum up $\sigma_i^2(g_m)$ to get $s_d^2(g_m)$. First, we have the decomposition:
\#\label{variance_decomposition}
s_d^2(g_m)&=\frac{1}{K}\sum_{1\le i\le K}\sigma_i^2(g_m)
= \frac{1}{K}\sum_{1\le i\le K}\left\{ b_{g_m,g_m}(\alpha_i)-b_{g_m}^2(\alpha_i)\right\}\nonumber\\
&= \frac{1}{K}\sum_{1\le i\le K}b_{g_m,g_m}(\alpha_i)  - \frac{1}{K}\sum_{1\le i\le K}b_{g_m}^2(\alpha_i).
\#
For the first summation in \eqref{variance_decomposition}, it holds that
\$
\lim_{d\rightarrow \infty}\left\{ \sum_{1\le i\le K}\frac{1}{K}b_{g_m,g_m}(\alpha_i)
-\frac{d}{K}\sum_{1\le i\le K}f_{y,\rho_{\kappa}}(\alpha_i)(\alpha_i-\alpha_{i+1}) b_{g_m,g_m}(\alpha_i)\right\} =0, 
\$
due to the fact that 
\$
\lim_{d\rightarrow \infty}\frac{f_{y,\rho_{\kappa}}(\alpha_i)(\alpha_i-\alpha_{i+1})}{d^{-1}}
= \lim_{d\rightarrow \infty}\frac{f_{y,\rho_{\kappa}}(\alpha_i)(\alpha_i-\alpha_{i+1})}{ F_{y,\rho_\kappa}(\alpha_{i}) - F_{y,\rho_\kappa}(\alpha_{i+1})}
= \lim_{d\rightarrow \infty}\frac{f_{y,\rho_{\kappa}}(\alpha_i)(\alpha_i-\alpha_{i+1})}{ \int_{\alpha_{i+1}}^{\alpha_i}f_{y,\rho_{\kappa}}(t)\mathrm{d}t}
= 1
\$ 
for each $1\le i\le K$, where the last step is implied by the continuity of $f_{y,\rho_{\kappa}}$ and the fact that $|\alpha_i-\alpha_{i+1}|=o(1)$.  
Furthermore, we have the limit 
\$
\lim_{d\rightarrow \infty}\sum_{1\le i\le d}f_{y,\rho_{\kappa}}(\alpha_i)(\alpha_i-\alpha_{i+1}) b_{g_m,g_m}(\alpha_i) 
= \int_{t>0} f_{y,\rho_{\kappa}}(t)b_{g_m,g_m}(t) \mathrm{d}t 
\equiv \int_{t>0} b_{g_m,g_m}(t)\mathrm{d}F_{y,\rho_{\kappa}}(t),
\$
which leads to 
\$
\lim_{d\rightarrow \infty}\frac{1}{K}\sum_{1\le i\le K}b_{g_m,g_m}(\alpha_i) 
=(1\vee y)\int_{t>0} b_{g_m,g_m}(t) \mathrm{d} F_{y,\rho_{\kappa}}(t).
\$
Similarly, for the second summation in \eqref{variance_decomposition}, it holds that
\$
\lim_{d\rightarrow \infty}\frac{1}{K}\sum_{1\le i\le K}b_{g_m}^2(\alpha_i) 
=(1\vee y) \int_{t>0} b_{g_m}^2(t) \mathrm{d} F_{y,\rho_{\kappa}}(t).
\$
Therefore, we have 
\$
\lim_{d\rightarrow \infty} s_d^2(g_m)= (1\vee y)\left\{\int_{t>0} b_{g_m,g_m}(t) \mathrm{d} F_{y,\rho_{\kappa}}(t)-\int_{t>0} b_{g_m}^2(t) \mathrm{d} F_{y,\rho_{\kappa}}(t)\right\}
=: s^2(g_m).
\$

{\bf Step 5}. Calculate the mean $K^{-1}\sum_{1\le i\le K}\mathbb E \left\{g_m(\alpha_i+\ell_i)\right\}$. First,
\$
\mathbb E \left\{g_m(\alpha_i+\ell_i)\right\}=\int g_m(\alpha_i+t) \mathrm{d} H(t) 
= \int g(\alpha_i+t)h(t) \mathrm{d}t = b_{g_m}(\alpha_i).
\$
Adopting the similar calculation of $d^{-1}\sum_{1\le i\le d}b_{g_m,g_m}(\alpha_i)$ in {\bf Step 4}, we have
\$
\lim_{d\rightarrow \infty} \frac{1}{K}\sum_{1\le i\le K}\mathbb E \left\{g_m(\alpha_i+\ell_i)\right\}=
\lim_{d\rightarrow \infty}\frac{1}{K}\sum_{1\le i\le K}{b_{g_m}(\alpha_i)} = (1\vee y)\int_{t>0} {b_{g_m}(t)}\mathrm{d} F_{y,\rho_{\kappa}}(t).
\$

{\bf Step 6}. Conclude the CLT. According to the CLT in Lemma \ref{Lyapunov_CLT}, the facts $s_d(g_m)\rightarrow s(g_m)$ and $K^{-1}\sum_{1\le i\le K}\mathbb E g_m(\alpha_i+\ell_i)\rightarrow (1\vee y) \int_{t>0} b_{g_m}(t)\mathrm{d}F_{y,\rho_{\kappa}}(t)=\mu_{\kappa}(g_m)$, we conclude that 
\$
\frac{\sqrt K}{s(g_m)}\left\{ L^{\mathrm{dp}}(g_m) - \mu_{\kappa}(g_m)\right\}
\xrightarrow{D} \mathcal N (0,1).
\$

\vspace{1em}

\noindent{\bf Part B.} We next deal with the asymptotic distribution for $L^{\mathrm{dp}}(\mathbf{v}^{\top}\mathbf{g})$. We first show the convergence of the covariance matrix of $L^{\mathrm{dp}}(\mathbf{g})$. 
It suffices to derive the asymptotic correlation between any two components $L^{\rm dp}(g_m)$ and $L^{\rm dp}(g_s)$ for $1\le m,s\le r$. We centralize and rescale $L^{\rm dp}(g_m)$ as $\tilde{L}^{\mathrm{dp}}(g_m) = \sqrt{K} \left\{{L}^{\mathrm{dp}}(g_m) - \mu_{\kappa}(g_m) \right\}$, $1\le m\le r$. Then
\#\label{cov_decomposition}
{\rm Cov} \left\{\tilde{L}^{\mathrm{dp}} (g_m),\tilde{L}^{\mathrm{dp}} (g_s)\right\}
& = \EE \left\{ K\cdot L^{\mathrm{dp}} (g_m) L^{\mathrm{dp}} (g_s)\right\} - K \mu_{\kappa}(g_m) \mu_{\kappa}(g_s),
\#
where the expectation operator $\EE(\cdot)$ is taken with respect to the randomness in both $\{\lambda_i\}_{i=1}^K$ and $\{\ell_i\}_{i = 1}^K$. 

By the same arguments used in {\bf Step 1} above, we make an approximation for the first term in the RHS of \eqref{cov_decomposition} as
\#\label{cov_approximation}
\EE \left\{ K\cdot L^{\mathrm{dp}} (g_m) L^{\mathrm{dp}} (g_s)\right\}
= & \EE \left\{\frac{1}{K}\sum_{1\le i,j\le K}g_m(\lambda_i+\ell_i)g_s(\lambda_j+\ell_j)\right\}\nonumber\\
\approx &  \EE \left\{\frac{1}{K}\sum_{1\le i,j\le K}g_m(\alpha_i+\ell_i)g_s(\alpha_j+\ell_j)\right\},
\#
where the approximation error induced by ``$\approx$'' is of order $o(K^{-1/2})$. 
Further, we have
\#\label{cov_convergence}
&\EE \left\{\frac{1}{K}\sum_{1\le i,j\le K}g_m(\alpha_i+\ell_i)g_s(\alpha_j+\ell_j)\right\}\nonumber\\
= & \EE \left\{\frac{1}{K}\sum_{1\le i\neq j\le K}g_m(\alpha_i+\ell_i)g_s(\alpha_j+\ell_j)\right\}+ \EE \left\{\frac{1}{K}\sum_{1\le i\le  K}g_m(\alpha_i+\ell_i)g_s(\alpha_i+\ell_i)\right\}\nonumber\\
= & \frac{1}{K}\sum_{1\le i\neq j\le K}\EE \left\{g_m(\alpha_i+\ell_i)\right\} \EE \left\{ g_s(\alpha_j+\ell_j)\right\} + \frac{1}{K}\sum_{1\le i\le K}\EE \left\{ g_m(\alpha_i+\ell_i)g_s(\alpha_i+\ell_i)\right\}\nonumber\\
= & \frac{1}{K}\sum_{1\le i, j\le K}\EE \left\{g_m(\alpha_i+\ell_i)\right\} \EE \left\{ g_s(\alpha_j+\ell_j)\right\} - \frac{1}{K}\sum_{1\le i\le K}\EE \left\{g_m(\alpha_i+\ell_i)\right\} \EE \left\{ g_s(\alpha_i+\ell_i)\right\}\nonumber\\
& + \frac{1}{K}\sum_{1\le i\le K}\EE \left\{ g_m(\alpha_i+\ell_i)g_s(\alpha_i+\ell_i)\right\}\nonumber\\
= & \frac{1}{K}\sum_{1\le i\le K} b_{g_m}(\alpha_i)\sum_{1\le j\le K} b_{g_s}(\alpha_j) - \frac{1}{K}\sum_{1\le i\le K}b_{g_m}(\alpha_i) b_{g_s}(\alpha_i)  + \frac{1}{K}\sum_{1\le i\le K} b_{g_m, g_s}(\alpha_i)\nonumber\\
\rightarrow & K (1\vee y)\int_{t>0} b_{g_m}(t){\rm d} F_{y,\rho_{\kappa}} (t) (1\vee y)\int_{t>0} b_{g_s}(t){\rm d} F_{y,\rho_{\kappa}} (t) \nonumber\\
&- (1\vee y)\int_{t>0} b_{g_m}(t) b_{g_s}(t){\rm d} F_{y,\rho_{\kappa}} (t) + (1\vee y)\int_{t>0} b_{g_m, g_s}(t) {\rm d} F_{y,\rho_{\kappa}} (t)\nonumber\\
= & K \mu_{\kappa}(g_m) \mu_{\kappa}(g_s) + (1\vee y)\int_{t>0} \left\{b_{g_m,g_s}(t) - b_{g_m}(t) b_{g_s}(t)\right\}{\rm d} F_{y,\rho_{\kappa}} (t),
\#
where the convergence ``$\rightarrow$'' in the fifth step is achieved similar to the calculation of the RHS of \eqref{variance_decomposition} in {\bf Step 4}.

Thus, \eqref{cov_decomposition}, \eqref{cov_approximation} and \eqref{cov_convergence} together lead to 
\$
{\rm Cov} \left\{\tilde L^{\mathrm{dp}} (g_m),\tilde L^{\mathrm{dp}}(g_s)\right\} \rightarrow (1\vee y)\int_{t>0} \left\{b_{g_m ,g_s}(t) - b_{g_m}(t) b_{g_s}(t)\right\}{\rm d} F_{y,\rho_{\kappa}} (t)=:v_{\kappa}(g_m,g_s).
\$
Accordingly, we write the limit of covariance matrix of $\sqrt{K}\left\{L^{\mathrm{dp}}(\mathbf{g}) - \mu_{\kappa}(\mathbf{g})\right\}$ as $V_\kappa (\mathbf g):= [v_\kappa (g_m,g_s)]_{m,s=1}^r$. \\

\noindent {\bf Part C.} To demonstrate the asymptotic normality of $L^{\rm dp}(\mathbf v^\top \mathbf g)$, it remains to verify that $\mathbf v^\top \mathbf g$ also satisfies all assumptions imposed on each $g_m$. The regularity~\eqref{approximately equal} in Assumption~\ref{assumption approximation error} holds obviously, so it suffices to check the Lyapunov's condition in Assumption~\ref{Lyapunov assumption}. 

Suppose there exists some constant $\delta>0$ such that the Lyapunov's condition \eqref{lyapunov} holds for $g_m, 1\leq m \leq r$. Now we verify the Lyapunov's condition for $\mathbf{v^{\top}g}= \sum_{m = 1}^rv_mg_m $, that is, 
\#\label{lyapunov_condition_vg}
\lim_{K \to \infty} \frac{1}{B_{K, \mathbf{v^\top g} }^{2+\delta}} \sum_{1\leq i \leq K} \mathbb E\left[ \left|\mathbf{v^{\top}g}(\theta_i + \ell_i) - \mathbb E \left\{\mathbf{v^{\top}g}(\theta_i + \ell_i)\right\} \right|^{2+\delta}\right] = 0,
\#
where $B_{K, \mathbf{v^{\top}g}}:= \sum_{1\leq i \leq K}\mathrm{Var} \{\mathbf{v^{\top}g}(\theta_i+ \ell_i)\}$.

First, using the inequality that $|\sum_{m=1}^r a_m|^{1+\nu} \leq r^{\nu}\sum_{m=1}^r|a_m|^{1+\nu} $ for any $\{a_m\}$ and $\nu >0$, the numerator in \eqref{lyapunov_condition_vg} for any $\mathbf{v^{\top}g}$ can be bounded by
\begin{align*}
&\quad \sum_{1\leq i\leq K} \mathbb E \left\{|\mathbf{v^{\top}g}(\theta_i+ \ell_i) - \mathbb E\mathbf{v^{\top}g}(\theta_i+ \ell_i)|^{2+\delta} \right\} \\&\leq r^{1+\delta } \sum_{1\leq m\leq r}v_m^{2+\delta}\sum_{1\leq i\leq K} \mathbb E \left\{ | {g_m}(\theta_i+ \ell_i) - \mathbb Eg_m(\theta_i+ \ell_i)|^{2+\delta} \right\}.
\end{align*}
According to Lyapunov's condition on $g_m$, $1\leq m \leq r$, we have
\$
\sum_{1\leq i\leq K} \mathbb E \left\{ | {g_m}(\theta_i+ \ell_i) - \mathbb Eg_m(\theta_i+ \ell_i)|^{2+\delta} \right\} 
= o\left(B_{K,m}^{2+\delta}\right) = o\left(K^{\frac{2+\delta}{2}}\right),\quad \forall 1 \leq m\leq r.
\$
This leads to 
\#\label{lyapunov_vg1}
\sum_{1\leq i\leq K} \mathbb E \left\{ | \mathbf{v^{\top}g}(\theta_i+ \ell_i) - \mathbb E\mathbf{v^{\top}g}(\theta_i+ \ell_i)|^{2+\delta} \right\} 
= o\left( r^{1+\delta}\cdot r\cdot \|\mathbf v\|_1^{2+\delta} \cdot K^{\frac{2+\delta}{2}}\right) = o\left(K^{\frac{2+\delta}{2}}\right).
\#
For the denominator in \eqref{lyapunov_condition_vg}, it holds that
\#\label{lyapunov_vg2}
B_{K, \mathbf{v^{\top}g}}^{2+\delta} &=  \left[\sum_{1\leq i \leq K}\mathrm{Var} \{\mathbf{v^{\top}g}(\theta_i+ \ell_i) \}\right]^{\frac{2+\delta}{2}} = \left(\mathbf{v}^{\top}\left[\sum_{1\leq i \leq K}\mathrm{Cov} \{\mathbf{g}(\theta_i+ \ell_i) \}\right]\mathbf{v} \right)^\frac{2+\delta}{2}\nonumber \\
& \gtrsim \left\{ K \lambda_r \left( V_\kappa (\mathbf g)\right) \right\}^{\frac{2+\delta}{2}}
= K^{\frac{2+\delta}{2} } \left\{\lambda_r \left( V_\kappa (\mathbf g)\right) \right\}^{\frac{2+\delta}{2}}.
\# 
So \eqref{lyapunov_vg1} and \eqref{lyapunov_vg2} imply the Lyapunov's condition for $\mathbf{v^{\top}g}$ once the the covariance matrix $V_{\kappa}(\mathbf{g})\in \RR^{r\times r}$ is non-degenrated, i.e., $\lambda_r (V_\kappa (\mathbf g) )>0$. 

Thus, this concludes the proof of Lemma~\ref{asymptotic normality}.

\subsection{Proof of Theorem~\ref{asymptotic normality 2}}

It suffices to verify that Assumptions \ref{Lyapunov assumption}-\ref{assumption approximation error} hold for $g_m, 1\leq m \leq 3$, where $g_1(x) = |x| - \log|x| -1$, $g_2(x) = (x-1)^2$ and $g_3(x) = |x-1|$. Let $\ell$ denote  the centered Laplacian variable with the probability density function $h(x)= \frac{1}{2\sigma}e^{-|x|/\sigma}$. 

\textbf{Step 1.}
For $1\leq m \leq 3$,
 if there exists some constants $c_1,c_2 >0$ such that
\begin{enumerate}
\item [(1)] $\inf_{\theta \in [a,b]} \mathbb E |g_m(\ell+ \theta) - \mathbb E g_m(\ell+ \theta)|^2 \geq c_1$,
    
\item[(2)] $\sup_{\theta \in [a,b]} \mathbb E |g_m(\ell+ \theta) - \mathbb E g_m(\ell+ \theta)|^4 \leq  c_2$,
\end{enumerate}
then Lyapunov's condition (\ref{lyapunov}) is satisfied for $g_m(x),1\leq m \leq 3$ by choosing $\delta = 2$. 

We first verify conditions (1) and (2) for $g_1(x) = |x| -\log|x| - 1$. Let $s_{k,1}(\theta)$ denote the $k$th moment of $g_1(\ell + \theta)$ for $k \geq 1$. Now we show the continuity of $s_{k,1}(\theta)$ w.r.t. $\theta \in [a,b]$ with $0\leq a<b$. Simple calculation yields that
$$
\begin{aligned}
s_{k,1}(\theta)& = \int_{-\infty}^{\infty}(|x+\theta|-\log|x+\theta|-1)^k  h(x)\mathrm{d}x= \int_{-\infty}^{\infty}(|x|-\log|x|-1)^k \cdot \frac{1}{2\sigma}e^{-|x-\theta|/\sigma}\mathrm{d}x\\
& = \frac{1}{2\sigma} \left\{e^{-\theta/\sigma} \int_{-\infty}^{\theta}(|x|-\log|x|-1)^k  e^{x/\sigma}\mathrm{d}x+ e^{\theta/\sigma}\int_{\theta}^{\infty}(|x|-\log|x|-1)^k e^{-x/\sigma}\mathrm{d}x\right\}.
\end{aligned}
$$
For the two improper integrals, by the convergence of $ \int_{-\infty}^{0}(|x|-\log|x|-1)^k \cdot\exp\left(\frac{x} {\sigma}\right)\mathrm{d}x$ and $ \int_0^{\infty}(|x|-\log|x|-1)^k \cdot\exp\left(-\frac{x} {\sigma}\right)\mathrm{d}x$, the continuity of $s_{k,1}(\theta)$ reduces to the continuity of $I_1(\theta) := \int_{0}^{\theta}(x-\log x-1)^k  e^{x/\sigma}\mathrm{d}x $ and $I_2(\theta):=\int_{0}^{\theta}(x-\log x-1)^k e^{-x/\sigma}\mathrm{d}x $ at the origin. By the convexity of $\phi(y) = y^k$ for $k\geq 1$ and the monotony of $\psi(y) = e^y$, we have
$$
|I_1(\theta)| = \left| \int_{0}^{\theta}(x-\log x-1)^k  e^{x/\sigma}\mathrm{d}x  \right|\leq e^{\theta/\sigma}  3^{k-1} \left\{ \int_{0}^{\theta} (x^k + |\log x|^k +1) \mathrm{d}x\right\},
$$
and similarly,
$$
|I_2(\theta)|= \left| \int_{0}^{\theta}  (x-\log x -1)^k e^{-x/\sigma}\mathrm{d}x \right|\leq  3^{k-1} \left\{ \int_{0}^{\theta} (x^k + |\log x|^k +1 ) \mathrm{d}x\right\}.
$$
Using the fact that $\int_0^{\theta}|\log x|^k\mathrm{d}x = (-1)^k\int_0^{\theta}(\log x)^k\mathrm{d}x$ for $0<\theta<1$ and the convergence of the improper integral $\int_0^{\theta}(\log x)^k\mathrm{d}x$, we have
$$
\lim_{\theta \to 0^{+}} I_1(\theta) = \lim_{\theta \to 0^{+}} I_2(\theta)= 0.
$$
We conclude that $s_{k,1}(\theta)$ is continuous w.r.t. $\theta \in [a,b]$ when $0 \leq a < b$, which implies there exists some $\theta_1^{\ast},\theta_2^{\ast} \in [a,b]$ such that

$$
\begin{aligned}
    &\inf_{\theta \in [a,b]} \mathbb E |g_1(\ell+ \theta) - \mathbb E g_1(\ell+ \theta)|^2=\mathbb E |g_1(\ell+ \theta_1^{\ast}) - \mathbb E g_1(\ell+ \theta_1^{\ast})|^2=:c_1,\\
    &\sup_{\theta \in [a,b]} \mathbb E |g_1(\ell+ \theta) - \mathbb E g_1(\ell+ \theta)|^4=\mathbb E |g_1(\ell+ \theta_2^{\ast}) - \mathbb E g_1(\ell+ \theta_2^{\ast})|^4=:c_2.
\end{aligned}
$$
For any $\theta \in [a,b]$, $g_1(\ell + \theta)$ is a substantially non-degenerate variable so the constants $c_1,c_2$ are strictly positive, which implies the Lyapunov's condition holds for $g_1$.

    Then we verify conditions (1) and (2) for $g_2(x) =(x-1)^2 $ and $g_3(x) = |x - 1|$. It reduces to show the continuity of $k$th moment of $g_m(\ell + \theta),m = 2,3$ for any $k \geq 1$. Denote them by $s_{k,2}(\theta)$ and $s_{k,3}(\theta)$, respectively. By definition,
$$
\begin{aligned}
s_{k,2}(\theta)& = \int_{-\infty}^{\infty}|x+\theta-1|^{2k}  h(x)\mathrm{d}x= \int_{-\infty}^{\infty}|x-1|^{2k} \cdot \frac{1}{2\sigma}e^{-|x-\theta|/\sigma}\mathrm{d}x\\
& = \frac{1}{2\sigma} \left\{e^{-\theta/\sigma} \int_{-\infty}^{\theta}|x-1|^{2k}  e^{x/\sigma}\mathrm{d}x+ e^{\theta/\sigma}\int_{\theta}^{\infty}|x - 1|^{2k} e^{-x/\sigma}\mathrm{d}x\right\},
\end{aligned}
$$
and 
$$
\begin{aligned}
s_{k,3}(\theta)& = \int_{-\infty}^{\infty}|x+\theta-1|^{k}  h(x)\mathrm{d}x= \int_{-\infty}^{\infty}|x-1|^{k} \cdot \frac{1}{2\sigma}e^{-|x-\theta|/\sigma}\mathrm{d}x\\
& = \frac{1}{2\sigma} \left\{e^{-\theta/\sigma} \int_{-\infty}^{\theta}|x-1|^{k}  e^{x/\sigma}\mathrm{d}x+ e^{\theta/\sigma}\int_{\theta}^{\infty}|x - 1|^{k} e^{-x/\sigma}\mathrm{d}x\right\}.
\end{aligned}
$$
Then the continuity of $s_{k,2}(\theta)$ and $s_{k,3}(\theta)$  w.r.t. $\theta \in [a,b]$ with $0\leq a<b$ follows from the convergence of associated improper integrals.

\textbf{Step 2}.
Now we turn to verify the regularity in Assumption~\ref{assumption approximation error} holds for $g_m, 1\leq m \leq 3$.

 For $m = 1$,
			$$
		\begin{aligned}
			\sum_{1\leq i \leq K} \left\{ g_1(\eta_i+ \ell_i) - g_1(\xi_i + \ell_i)\right\} 
			& = \sum_{1\leq i \leq K} \log \left|1+ \frac{ \xi_i - \eta_i}{\eta_i + \ell_i} \right| + \sum_{1\leq i \leq K}  \left(\left|\eta_i + \ell_i\right| - \left|\xi_i + \ell_i\right| \right)  \\
			&\leq 	\sum_{1\leq i \leq K}\left| \frac{ \xi_i - \eta_i}{\eta_i + \ell_i} \right| + \sum_{1\leq i \leq K}  \left| \xi_i -\eta_i\right|,
		\end{aligned}
		$$
		where the last equality follows using $\log|1+x| \leq |x|$ for $x \in \mathbb R$ and $|a| - |b| \leq |a-b|$ for $a , b \in \mathbb R$. Similarly, we also have
        $$
        \sum_{1\leq i \leq K} \left\{g_1(\xi_i + \ell_i) -  g_1(\eta_i+ \ell_i) \right\} 
			\leq 	\sum_{1\leq i \leq K}\left| \frac{ \xi_i - \eta_i}{\xi_i + \ell_i} \right| + \sum_{1\leq i \leq K}  \left| \xi_i -\eta_i\right|.
        $$
        As $|\xi_i - \eta_i | = O_p(K^{-1/2-\epsilon})$ for $1\leq i\leq K$, we have
		$$
		\frac{1}{K}	\sum_{1\leq i \leq K} \left\{ g_1(\xi_i+ \ell_i) - g_1(\eta_i + \ell_i)\right\} = O_p(K^{-3/2-\epsilon})\cdot \left(1+\sum_{1\leq i \leq K} \frac{ 1}{\left|\ell_i + \beta_i \right|} \right),
		$$
        where $\beta_i $ denotes $\xi_i$ or $\eta_i$ for $1\leq i\leq K$, and thus bounded almost surely.
		Now we concentrate on the divergence rate of $\sum_{1\leq i\leq K} U_i$ with $U_i = 1/|\ell_i + \beta_i|$.  For any $1\leq i\leq K$, by the exponential decay of $h(x)$, the survivor function of $U_i$ is
		$$
		\bar{F}_{U_i}(u):=\mathbb P(U_i>u)  \asymp \frac{c_i}{u},\qquad u \to \infty,
		$$
        where $c_i>0, 1\leq i\leq K$. 
	By the generalized central limit theorem (GCLT) for independent and non-identical random variables~\citep{shintani2018super}, there exists $a_K \asymp K\log K, b_K \asymp K$ and a stable law $H_{s}$ such that
		$$
		\frac{1}{b_K}\left(\sum_{1\leq i \leq K} U_i -a_K \right) \xrightarrow{D} H_{s}
		$$
		which means that 
		$$
			\frac{1}{K}	\sum_{1\leq i \leq K} \left\{ g_1(\xi_i+ \ell_i) - g_1(\eta_i + \ell_i)\right\} = O_p(K^{-3/2-\epsilon})\cdot O_p(K\log K) = o_p\left(K^{-1/2}\right).
		$$
For $m = 2$, 
\$
\frac{1}{K}\sum_{1\leq i \leq K} \left\{ g_2(\xi_i+ \ell_i) - g_2(\eta_i + \ell_i)\right\} = \frac{1}{K}	\sum_{1\leq i \leq K} \left(\xi_i - \eta_i\right)\left(\xi_i + \eta_i + 2 \ell_i -2\right) = o_p\left(K^{-1/2}\right),
\$
where the last equality follows by $|\xi_i - \eta_i | = O_p(K^{-1/2 -\epsilon})$ and $|\xi_i + \eta_i + 2 \ell_i -2|= O_p(1)$. For $m = 3$,  $g_3(x)$ is Lipschitz, condition (\ref{approximately equal}) holds. The proof is complete.

\subsection{Proof of Theorem~\ref{CLT_H1}}

Under Assumptions~\ref{assumption_Sigma_alternative}-\ref{regularity}, the support of the generalized M-P law is also contained in a finite closed interval $[a_1,b_1]$ with $0 \leq a_1 <b_1$. So the Lyapunov's condition holds under alternatives. Together with Lemma~\ref{asymptotic normality} and Theorem~\ref{asymptotic normality 2}, this completes the proof.

\section{Technical lemmas}
\label{app: B}
\begin{definition}[Stochastic dominance; Definition~3.4 in \citep{knowles2017anisotropic}]\label{stochastic_dominance} 
Let $\xi=\{\xi^{(n)}(u):n\in \NN, u\in U^{(n)}\}$ and $\eta=\{\eta^{(n)}(u):n\in\NN,u\in U^{(n)}\}$ be two families of nonnegative random variables with $U^{(n)}$ being a possibly $n$-dependent parameter set. If for any small constant $\epsilon>0$ and large constant $D>0$ it holds that 
\$
\sup_{u\in U^{(n)}}\PP \left\{\xi^{(n)}(u)> n^\epsilon \eta^{(n)}(u) \right\}\le n^{-D}
\$
for large enough $n\ge n_0 (\epsilon, D)$, then $\xi$ is said to be {\it stochastically dominated} by $\eta$ uniformly in $u$, denoted by $\xi\prec \eta$. 
\end{definition}

We define the {\it classical number} of eigenvalues in the $k$-th bulk component as
\#\label{def_dk}
d_k:= d \int_{a_{2k}}^{a_{2k-1} }\rho_{\kappa}(x)dx,
\#
which refers to the number of eigenvalues of $\mathbf{S_x}$ that located within the interval $[a_{2k},a_{2k-1}]$ in the limiting sense. Then we relabel the eigenvalues $\lambda_i$'s and classical locations $\alpha_i$'s in the $k$-th component as
\$
\lambda_{k,i}:=\lambda_{i+\sum_{\ell<k}d_\ell },\ 
\alpha_{k,i}:=\alpha_{i+\sum_{\ell<k}d_\ell}. 
\$

\begin{lemma}[Eigenvalue rigidity; Theorem~3.12 in~\citep{knowles2017anisotropic}]\label{rigidity}

Suppose that Assumption~\ref{regularity} holds. Then we have for all $1\leq k\leq p$ and $1\leq i \leq d_k$ satisfying $\alpha_{k,i}\ge \tau$ that 
\$ 
|\lambda_{k,i}-\alpha_{k,i}|\prec n^{-2/3}\left\{(i\wedge (d_k+1-i) \right\}^{-1/3},
\$
where ``$\prec$" denotes the stochastic dominance defined in Definition~\ref{stochastic_dominance}. 
\end{lemma}

\begin{lemma}[Lyapunov's central limit theorem]\label{Lyapunov_CLT}
Suppose $\{x_i\}_{i=1}^n$ is a sequence of independent random variables, each with finite mean $\mu_i$ and variance $\sigma_i^2$. Define $s_n^2=\sum_{1\le i\le n}\sigma_i^2$. If for some $\delta>0$, {\it Lyapunov's condition} 
\$
\lim_{n\rightarrow \infty}\frac{1}{s_n^{2+\delta}}\sum_{1\le i\le n}\mathbb E \left(|x_i-\mu_i|^{2+\delta}\right)=0
\$
is satisfied, then a sum of $(x_i-\mu_i)/s_n$ weakly converges to $\cN (0,1)$, that is,
\$
\frac{1}{s_n}\sum_{1\le i\le n}(x_i-\mu_i)\overset{D}{\rightarrow}\cN(0,1).
\$
\end{lemma}

\end{document}